\documentclass[3p,times]{elsarticle}

\usepackage{ecrc}

\volume{00}

\firstpage{1}

\journalname{Procedia Computer Science}

\runauth{}

\jid{procs}

\jnltitlelogo{Procedia Computer Science}

\CopyrightLine{2011}{Published by Elsevier Ltd.}

\usepackage{amssymb}

\usepackage[figuresright]{rotating}

\usepackage{afterpage}
\usepackage{makecell}
\usepackage{tabulary}
\usepackage{xcolor}
\usepackage{colortbl}
\usepackage{prettyref}
\usepackage{framed}
\usepackage{booktabs}       
\usepackage[utf8]{inputenc} 
\usepackage[T1]{fontenc}    
\usepackage{hyperref}       
\usepackage{url}            
\usepackage{booktabs}       
\usepackage{amsfonts}       
\usepackage{nicefrac}       
\usepackage{microtype}      
\usepackage{xcolor}         
\usepackage{amsthm}
\usepackage{amsmath}
\usepackage{amssymb}
\usepackage{algorithm}
\usepackage{algpseudocode}
\usepackage{graphicx}
\usepackage{subfigure}
\usepackage{booktabs} 
\usepackage{bm}
\usepackage{hyperref}
\usepackage{cleveref}
\usepackage{enumitem}
\usepackage{mathtools}
\usepackage{array}
\usepackage{lipsum}
\usepackage{multirow}
\usepackage{makecell}
\usepackage{xcolor}
\usepackage{tikz}
\usepackage{etoolbox}
\newcommand{\pref}[1]{\prettyref{#1}}

\newcommand{\savehyperref}[2]{\texorpdfstring{\hyperref[#1]{#2}}{#2}}
\newrefformat{eq}{\savehyperref{#1}{Eq. \textup{(\ref*{#1})}}}
\newrefformat{eqn}{\savehyperref{#1}{Eq.~\textup{(\ref*{#1})}}}
\newrefformat{lem}{\savehyperref{#1}{Lemma~\ref*{#1}}}
\newrefformat{assump}{\savehyperref{#1}{Assumption~\ref*{#1}}}
\newrefformat{defi}{\savehyperref{#1}{Definition~\ref*{#1}}}\newrefformat{tab}{\savehyperref{#1}{Table~\ref*{#1}}}
\newrefformat{lemma}{\savehyperref{#1}{Lemma~\ref*{#1}}}
\newrefformat{line}{\savehyperref{#1}{Line~\ref*{#1}}}
\newrefformat{thm}{\savehyperref{#1}{Theorem~\ref*{#1}}}
\newrefformat{corr}{\savehyperref{#1}{Corollary~\ref*{#1}}}
\newrefformat{cor}{\savehyperref{#1}{Corollary~\ref*{#1}}}
\newrefformat{sec}{\savehyperref{#1}{Section~\ref*{#1}}}
\newrefformat{app}{\savehyperref{#1}{Appendix~\ref*{#1}}}
\newrefformat{ex}{\savehyperref{#1}{Example~\ref*{#1}}}
\newrefformat{fig}{\savehyperref{#1}{Figure~\ref*{#1}}}
\newrefformat{alg}{\savehyperref{#1}{Algorithm~\ref*{#1}}}
\newrefformat{rem}{\savehyperref{#1}{Remark~\ref*{#1}}}
\newrefformat{conj}{\savehyperref{#1}{Conjecture~\ref*{#1}}}
\newrefformat{prop}{\savehyperref{#1}{Proposition~\ref*{#1}}}
\newrefformat{proto}{\savehyperref{#1}{Protocol~\ref*{#1}}}
\newrefformat{prob}{\savehyperref{#1}{Problem~\ref*{#1}}}
\newrefformat{claim}{\savehyperref{#1}{Claim~\ref*{#1}}}
\newrefformat{que}{\savehyperref{#1}{Question~\ref*{#1}}}
\newrefformat{op}{\savehyperref{#1}{Open Problem~\ref*{#1}}}
\newrefformat{fn}{\savehyperref{#1}{Footnote~\ref*{#1}}}
\newrefformat{property}{\savehyperref{#1}{Property~\ref*{#1}}}

\usepackage{xspace}
\usepackage{amsmath}
\usepackage{amsthm}
\usepackage{bbm}
\usepackage{mathrsfs}
\usepackage{bm}
\usepackage{makecell}
\usepackage{tabulary}

\usepackage[utf8]{inputenc}
\usepackage[english]{babel}

\DeclareMathOperator*{\argmin}{argmin} 
\DeclareMathOperator*{\argmax}{argmax} 

\newcommand{\calN}{\mathcal{N}}

\newcommand{\calM}{\mathcal{M}}

\newcommand{\calW}{\mathcal{W}}

\newcommand{\R}{\mathbb{R}}
\newcommand{\U}{{\mathcal{U}}}

\newcommand{\wtilde}{\widetilde}

\newcommand{\one}{\mathbbm{1}}

\newcommand{\blue}[1]{{\color{black}#1}}

\newtheorem*{theorem*}{Theorem}

\newtheorem{assumption}{Assumption}[section]

\newtheorem*{lemma*}{Lemma}
\newtheorem{definition}{Definition}[section]

\theoremstyle{definition}

\newcommand*{\circled}[1]{\lower.7ex\hbox{\tikz\draw (0pt, 0pt)%
    circle (.5em) node {\makebox[1em][c]{\small #1}};}}
\robustify{\circled}

\newtheorem{thm}{Theorem}[section]

\newtheorem{lem}[thm]{Lemma}

\usepackage{afterpage}
\usepackage{prettyref}
\usepackage{pgfplots}
\usepackage{caption}
\usepackage{amsmath,amssymb,amsfonts}
\usepgfplotslibrary{fillbetween}
\usepackage{framed}
\newrefformat{eq}{\savehyperref{#1}{Eq. \textup{(\ref*{#1})}}}
\newrefformat{eqn}{\savehyperref{#1}{Eq.~\textup{(\ref*{#1})}}}
\newrefformat{lem}{\savehyperref{#1}{Lemma~\ref*{#1}}}
\newrefformat{lemma}{\savehyperref{#1}{Lemma~\ref*{#1}}}
\newrefformat{defi}{\savehyperref{#1}{Definition~\ref*{#1}}}
\newrefformat{line}{\savehyperref{#1}{Line~\ref*{#1}}}
\newrefformat{thm}{\savehyperref{#1}{Theorem~\ref*{#1}}}
\newrefformat{corr}{\savehyperref{#1}{Corollary~\ref*{#1}}}
\newrefformat{cor}{\savehyperref{#1}{Corollary~\ref*{#1}}}
\newrefformat{sec}{\savehyperref{#1}{Section~\ref*{#1}}}
\newrefformat{table}{\savehyperref{#1}{Table~\ref*{#1}}}
\newrefformat{app}{\savehyperref{#1}{Appendix~\ref*{#1}}}
\newrefformat{assum}{\savehyperref{#1}{Assumption~\ref*{#1}}}
\newrefformat{ex}{\savehyperref{#1}{Example~\ref*{#1}}}
\newrefformat{fig}{\savehyperref{#1}{Figure~\ref*{#1}}}
\newrefformat{alg}{\savehyperref{#1}{Algorithm~\ref*{#1}}}
\newrefformat{rem}{\savehyperref{#1}{Remark~\ref*{#1}}}
\newrefformat{conj}{\savehyperref{#1}{Conjecture~\ref*{#1}}}
\newrefformat{prop}{\savehyperref{#1}{Proposition~\ref*{#1}}}
\newrefformat{proto}{\savehyperref{#1}{Protocol~\ref*{#1}}}
\newrefformat{prob}{\savehyperref{#1}{Problem~\ref*{#1}}}
\newrefformat{claim}{\savehyperref{#1}{Claim~\ref*{#1}}}
\newrefformat{que}{\savehyperref{#1}{Question~\ref*{#1}}}
\newrefformat{op}{\savehyperref{#1}{Open Problem~\ref*{#1}}}
\newrefformat{fn}{\savehyperref{#1}{Footnote~\ref*{#1}}}

\usepackage{xspace}
\usepackage{amsmath}
\usepackage{amsthm}
\usepackage{bbm}
\usepackage{mathrsfs}
\usepackage{amssymb}
\usepackage{bm}
\usepackage{makecell}
\usepackage{tabulary}

\theoremstyle{definition}

\begin{document}

\begin{frontmatter}



\renewcommand{\thefootnote}{\fnsymbol{footnote}}

\dochead{}

\title{No Free Lunch Theorem for Privacy-Preserving LLM Inference}






\author[label1]{Xiaojin Zhang
\footnote{National Engineering Research Center for Big Data Technology and System, Services Computing Technology and System Lab, Cluster and Grid Computing Lab, School of Computer Science and Technology, Huazhong University of Science and Technology, Wuhan, 430074, China, email: xiaojinzhang@hust.edu.cn}}
\address[label1]{Huazhong University of Science and Technology, China}
\author[label1]{Yahao Pang}
\author[label2]{Yan Kang}
\address[label2]{WeBank, China}
\author[label1]{Wei Chen}
\author[label2]{Lixin Fan}
\author[label1]{Hai Jin}
\author[label2,label3]{Qiang Yang\footnote{Corresponding Author, email: qyang@cse.ust.hk}}
\address[label3]{Hong Kong University of Science and Technology, China}

\begin{abstract}
Individuals and businesses have been significantly benefited by Large Language Models (LLMs) including PaLM, Gemini and ChatGPT in various ways. For example, LLMs enhance productivity, reduce costs, and enable us to focus on more valuable tasks. Furthermore, LLMs possess the capacity to sift through extensive datasets, uncover underlying patterns, and furnish critical insights that propel the frontiers of technology and science. However, LLMs also pose privacy concerns. Users' interactions with LLMs may expose their sensitive personal or company information. A lack of robust privacy safeguards and legal frameworks could permit the unwarranted intrusion or improper handling of individual data, thereby risking infringements of privacy and the theft of personal identities. To ensure privacy, it is essential to minimize the dependency between shared prompts and private information. Various randomization approaches have been proposed to protect prompts' privacy, but they may incur utility loss compared to unprotected LLMs prompting. Therefore, it is essential to evaluate the balance between the risk of privacy leakage and loss of utility when conducting effective protection mechanisms. The current study develops a framework for inferring privacy-protected Large Language Models (LLMs) and lays down a solid theoretical basis for examining the interplay between privacy preservation and utility. The core insight is encapsulated within a theorem that is called as the NFL (abbreviation of the word No-Free-Lunch) Theorem.
\end{abstract}


\end{frontmatter}





\section{Introduction}

The advent of sophisticated Large Language Models, including 
PaLM~\cite{chowdhery2022palm} and ChatGPT~\cite{Chatgpt} have brought substantial benefits to both individuals and enterprisess. These models are equipped to facilitate our endeavors across a diverse spectrum of domains, from synthesizing information to generating new content and data analysis. By doing so, they enhance our productivity, reduce costs, and free us from tedious work, allowing us to focus on more valuable tasks \cite{zhang2023toward}. Moreover, LLMs can assist in generating ideas, designing solutions, and facilitating research and development. For instance, in domains such as healthcare, finance, and science, LLMs can analyze massive volumes of data, identify patterns, and offer valuable insights that can propel technical and scientific breakthroughs and advancements.

\begin{figure}[thpb]
  \centering
  \includegraphics[width=0.7\linewidth]{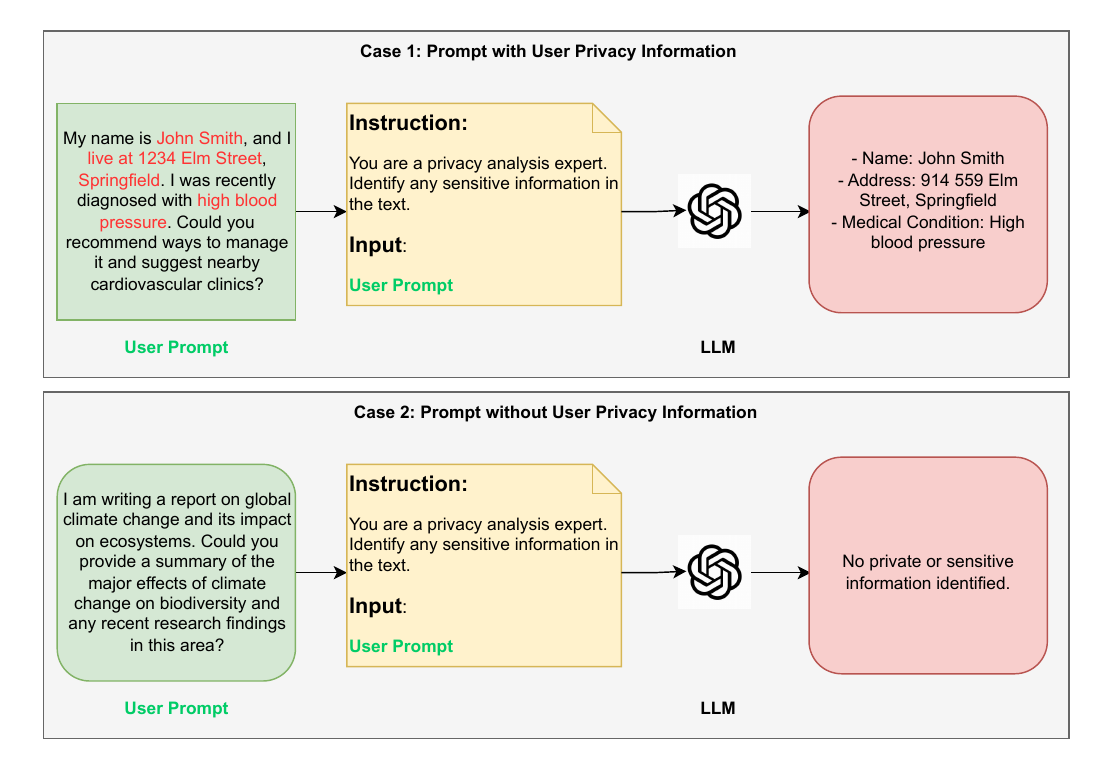}
  \caption{ \textbf{Privacy Leakage in Prompts.}
    In Case 1, the prompt contains sensitive user information, including: name (John Smith), address (1234 Elm Street, Springfield), and medical condition (high blood pressure). This enables the LLM to identify and extract private data, potentially leading to privacy leakage. Such information risks exposing the user’s identity, location, and health details, which could be misused for identity theft, targeted advertising, or discrimination.
    In contrast, the prompt in Case 2 excludes sensitive information, significantly reducing the risk of privacy leakage. This comparison highlights how sensitive information in user prompts can be easily identified by LLMs.
} 
  \label{prompt has privacy}
\end{figure}

While LLMs offer significant benefits, they also raise important privacy concerns. User prompts, which are central to interactions with LLMs, often include sensitive information such as personal identity, background details, and other confidential data, as illustrated in Figure \ref{prompt has privacy}. 
When users do not employ privacy protection measures, any prompt containing sensitive information can be easily recognized and extracted by the server. Even with certain privacy protection techniques in place, servers may still deploy various attack methods—such as embedding inversion or model inference—to infer sensitive user data. These threats have been substantiated in numerous studies (see Section \ref{sec:PAL}). However, overly stringent privacy protections can compromise the utility of LLMs, while insufficient protection leaves user privacy vulnerable. Therefore, our research aims to strike a balance by ensuring privacy protection for user prompts during the LLM inference phase, with minimal impact on model utility.

Therefore, preserving privacy is crucial when prompting LLMs for inference. In this context, a natural approach to safeguarding the privacy of prompts is to introduce randomization. 

This study focuses on a typical LLM inference setting where the LLM operates as a black box, concealing its architecture, model parameters, and inference details within commercial APIs and user interfaces. To address this issue, several randomization techniques~\cite{zhou2023textobfuscator,tong2023infer,li2023privacy} have been proposed, demonstrating empirical effectiveness in mitigating privacy leakage resulting from exposed prompts. However, these approaches inevitably incur a certain degree of utility loss compared to prompting LMs without any protection. Moreover, existing research lacks the theoretical analysis necessary to quantify the compromise of the reduction in utility and exposure of privacy. The tension between reducing the exposure of private information and the detriment to utility propels our exploration to investigate the research inquiry: "\textit{Theoretically, is there possibility of developing some protective methods that can achieve minimal exposure of private information and detriment to utility at the same time when prompting an LLM?}" Our primary finding is encapsulated in the Theorem \ref{thm: utility-privacy trade-off_mt} within context of our proposed framework for privacy-preserving LM inference, presenting a counterargument to this question. We named it as No-Free-Lunch Theorem.

Another pertinent research~\cite{zhang2022no} similarly introduces a theorem that try to address the equilibrium between the utility and privacy. However, the scope of that study primarily focuses on training models in the context of horizontal federated learning, where privacy leakage is defined based on model gradients transmitted between clients and the server. In contrast, our work concentrates on LLM inference, where interactions involving prompts and LLM outputs occur between the client and server as detailed in section \ref{sec:llm_framework}, and we define privacy leakage as the extent to which an adversary can deduce the original input from an protected one. These distinct definitions of privacy leakage provide novel theoretical insights for analyzing the balance between utility and private information. Our study's contributions are encapsulated in the following points:

\begin{itemize}
    \item A framework aimed at privacy-preserving LLM inference, wherein we provide formal definitions for privacy leakage (as outlined in Definition \ref{defi: privacy_leakage}) and utility loss (as outlined in Definition \ref{defi: utility_loss}), has been proposed. Furthermore, we formulate the privacy-preserving LLM inference problem as a constrained optimization problem.
    \item A No-Free-Lunch theorem (as stated in Theorem \ref{thm: utility-privacy trade-off_mt}) aimed at privacy-preserving LLM inference, which provides a quantitative characterization of the trade-off between detriment to utility and exposure of private information, has been established. Specifically, it demonstrates that weighted summation of utility reduction and privacy exposure is bounded below by a constant contingent upon the specific problem, and is not zero, thereby implying that a certain degree of model performance must be sacrificed in order to ensure a desired level of privacy preservation.
\end{itemize}

\section{Related Work}


During the LLM inference, both the client and the LLM server have the potential to act as adversaries, compromising each other's privacy~\cite{kang2023grounding, zhu2012discovering}. In this study, we focus on a scenario where the LLM server is the adversary and can infer the client's private information by analyzing the prompt provided by the client (see Section \ref{sec:tm}). In this section, we provide a concise overview of existing research concerning attack of privacy, safeguarding measures, and the balance of preserving utility against protecting privacy within the realm of large language model (LLM) inference.

\subsection{Privacy Attacks in LLM inference}\label{sec:PAL}

During the inference phase of LLM, the LLM server may attempt to infer a client's private information by analyzing the prompts sent by the client. This can be achieved by inferring privacy information from a single prompt or by manipulating an otherwise innocuous conversation with the user to elicit prompts that contain private and sensitive information~\cite{staab2023beyond}. A multitude of strategies for launching attacks have been outlined in such contexts. For example, the server has the capability to initiate attacks that invert embeddings~\cite{qu2021natural, morris2023text, li2023sentence, kugler2021invbert}, aiming to reconstruct the original prompt based on the provided embedding. In addition, the server may employ these kinds of attacks for extracting sensitive information such as ethnicity, sex, and age from the provided embedding~\cite{Song-2020-Information, li-etal-2022-dont, song2019overlearning, hayet2022invernet, yue2021differential}. Moreover, the server hosting the LLM can exploit its model to discern the client's original prompt from a modified version~\cite{tong2023infer, staab2023beyond}.

\subsection{Privacy Protections in LLM inference}

Secure Multi-Party Computation (SMPC) and Randomization are two mainstream protection mechanisms leveraged to protect the privacy of clients' sensitive information. SMPC facilitates the collaborative computation of a function among several entities, ensuring the privacy of their individual inputs to be confidential. In the context of LLMs, SMPC could be employed for executing computations on encrypted data, which ensures that no single party gains access to other parties' confidential inputs. These methods prioritize enhancing the efficiency of LLM architectures and SMPC protocols to minimize the substantial computation and communication expenses incurred by SMPC protocols~\cite{li2023mpcformer, dong2023puma,liu2023llms,zheng2023primer,hou2023ciphergpt, hao2022iron, ding2023east}. Although SMPC can ensure the confidentiality of client data, it requires the LLM to collaborate tightly with the client, which can limit its application to LLMs that are only accessible through commercial APIs. Randomization is an alternative privacy-preserving mechanism that adds random noise to the prompt's embedding to protect the privacy of the prompt. For instance, the InferDPT approach~\cite{tong2023infer} utilizes a differential privacy (DP) mechanism to alter the user's input text, thereby hindering any potential eavesdropping by a malevolent large language model (LLM) server that could deduce sensitive user information. In the context of the InferDPT system, users initially introduce minor, yet semantically coherent, variations to the tokens within their input. Afterward, this modified input is forwarded to an LLM, which formulates a reply and relays it back to the user. The user then utilizes their own pre-trained model to produce the conclusive output by integrating the initial input and the LLM's reply. The DP-OPT technique~\cite{hong2023dpopt}, on the other hand, harnesses advanced Deep Language Networks (as described in DLN~\cite{sordoni2023dln}) under the guidance of a local model to autonomously refine input prompts. In this process, DP-OPT implements a privacy-preserving aggregation technique to safeguard the confidentiality of the prompts. In a separate study, staab et al.~\cite{staab2023beyond} explored the application of text anonymization techniques for safeguarding user data privacy and concluded through empirical evidence that such methods alone are not adequate for ensuring robust privacy protection.

\subsection{Trade-off Between Utility and  Privacy in LLM inference}

In addition, beyond introducing methods to safeguard personal data, several investigations have scrutinized the equilibrium between utility and private information when interacting with LLMs. These investigations primarily focus on two scenarios for LLM inference. 

In the first scenario, the LLM operates behind a commercial API, solely relying on textual prompts as inputs. Protection mechanisms typically involve introducing random noise to the prompts or their corresponding embeddings. Representative works in this scenario include DP-OPT~\cite{hong2023dpopt} and InferDPT~\cite{tong2023infer}. These studies have demonstrated a decline in task performance as the privacy budget decreases. Furthermore, DP-OPT highlights that leveraging larger LLMs can significantly mitigate performance-private information trade-off. LLM is divided into two parts in the second scenario: a larger portion deployed on the server and a smaller portion deployed on the client. The server and client collaborate in training these two portions of the LLM, incorporating random noise into the embedding vectors to protect privacy. Representative examples of research falling under this scenario include TextObfuscator~\cite{zhou2023textobfuscator} and SAP~\cite{shen2023sap}. Both studies also exhibited the delicate balance between privacy and utility. Specifically, as randomization increases, utility decreases while privacy increases, and vice versa. 

The primary emphasis of these studies lies in the empirical assessment of the balance between confidentiality and performance, lacking a theoretical framework or a numerical evaluation of the utility-privacy equilibrium during LLMs' inference. Our study aims to address this gap by providing a theoretical analysis and quantification in this regard.

\begin{table*}[!htp]
\footnotesize
  \centering
  \setlength{\belowcaptionskip}{15pt}
  \caption{Used notation in this study}
  \bgroup
  \def\arraystretch{1.35}
  \setlength{\tabcolsep}{1.8pt}
  \label{table: notation}
    \begin{tabular}{cc}
    \toprule
    Notation & Meaning\cr
    \midrule\
    $\wtilde{d}^{(m)}$ & The $m$-th token in client's protected prompt $\wtilde{d}$ \cr
    $\epsilon_{p}$ & Leakage of private information (Def. \ref{defi: privacy_leakage})\cr
    $d^{(m)}$ & The $m$-th token in client's prompt $d$ \cr
    $\epsilon_u$ & Loss of utility (Def. \ref{defi: utility_loss}) \cr
    $d$, $\wtilde{d}$ & The client's original and protected prompts, respectively\cr
    $w$ & Undistorted embedding of the client's prompt\cr
    $\wtilde{w}$ & Distorted embedding of the client's prompt\cr
 $P$ & Distribution of undistorted embedding $w$\cr
 $\wtilde P$ & Distribution of distorted embedding $\wtilde{w}$\cr
$\breve P$ & Distribution of embedding that is independent of the embedding of the client's prompt\cr
  $P_0$ & Distribution of test data\cr
 $\text{TV}(\cdot||\cdot)$ & Two distributions' total variation distance\cr
    \bottomrule
    \end{tabular}
 \egroup
\end{table*}

\section{A Framework of Privacy-Preserving LLM Inference}\label{sec:llm_framework}

Our framework considers a typical LLM inference setting in which a client sends prompts to query the \textit{black-box} LLM hosted by a server. The LLM is black-box in the sense that the server of the LLM hides the LLM architecture and parameters as well as inference details, and it only exposes the query and prediction commercial APIs and interface for the client to make inferences. We also assume that the server may mount a privacy attack during the inference to infer client's privacy based on observed prompts, which necessitates client-side privacy protection. 

We commence with an introduction of threat model in this section. Subsequently, we expound upon representative protection mechanisms and attacking methods, based on which we subsequently provide formal definitions of utility loss and privacy leakage. Finally, we formulate objective of privacy-preserving LLM inference as a constrained optimization problem.


\subsection{Threat Model}\label{sec:tm}

We assume the LLM server is the attacker. We discuss its threat to the client's private data considering objective, capability, and knowledge of attacker. 


\noindent\textbf{Attacker's objective}. We regard the LLM server as a potential adversary intent on deducing the client's confidential details with considerable accuracy by analyzing the client's exposed embedding. To this end, the server tries to obtain tokens or words in the prompt as many as possible.


\noindent\textbf{Attacker's capability}. We classify the adversary as {semi-honest}, adhering to the LLM inference rules by ensuring the production of the output, yet potentially seeking to deduce sensitive client data based on client's prompt.

\noindent\textbf{Attacker's knowledge}. 
We assume that the server is aware that the client may apply a randomization mechanism to protect its uploaded prompts. The server may launch attacks using all available information (e.g., its hosted LLM) to reconstruct each word or token from the original prompt based on the observed perturbed embedding.

Additionally, while our primary focus is on scenarios where the LLM server acts as the adversary, we also consider the potential risks posed by malicious clients. Malicious clients may attempt to infer other clients’ private data by exploiting shared inference results, model updates, or interactions within collaborative learning frameworks. Our framework is designed to mitigate such risks by implementing a robust isolation mechanism, where each client’s prompt undergoes random perturbation before being processed by the LLM. This approach ensures that the LLM server, acting as an intermediary, does not directly access unprotected client data, thereby preventing malicious clients from inferring sensitive information about other users. As a result, our framework is adaptable to scenarios involving both a potentially adversarial server and malicious clients, ensuring comprehensive protection of client privacy.

\subsection{Attacking Methods}\label{sec:attack}

Considering the attacker's knowledge and capability, it is expected that they will employ a privacy breach technique, designated as Attack $\mathcal{A}$, to deduce the sequence of the original prompt's elements from the observed prompt. Herein, we give three representative attacks.
\begin{itemize}
    \item Input inference attack~\cite{yue2021differential}: the server leverages a pre-trained BERT model to reconstruct tokens of the original prompt from the perturbed prompt. Specifically, the server substitutes each token in the perturbed prompt sequentially with the special one "[MASK]" and lets the BERT model predict the token at the "[MASK]" position. The accuracy of the attack is computed by comparing the predicted tokens to the original tokens. 

 \item Embedding inversion attack~\cite{qu2021natural}: the server deduces the original tokens by referencing the embeddings of the protected tokens. It utilizes an algorithm for nearest neighbor identification to accomplish this objective. Upon receiving a embedding of an altered token, the server identifies its closest counterpart within the vector space, which is then considered the inferred original token. The effectiveness of this process is gauged by the precision of the recovered tokens.

 \item LLM-assisted inference attack~\cite{tong2023infer,staab2023beyond}: the server exploits the capitalizes of large language models like GPT-4 to reconstruct the original prompt from the perturbed version. To elaborate, the server introduces the perturbed prompt to the GPT-4 system and commands it to regenerate each token. The operation is deemed a success when the regenerated tokens align precisely with those of the original prompt.
\end{itemize}

In this study, we investigate the trade-off between loss of utility and leakage of privacy when client introduces randomization to the prompt used to query a black-box LLM for protecting the prompt's confidentiality. In the next section, we elaborate on the randomization protection mechanism and the associated privacy-preserving LLM inference process.


\subsection{Protection Mechanisms}\label{sec:protect}



The  user employs \textit{privacy preserving mechanism} $\calM: \R^{m}\rightarrow \R^{m}$ that convert the embedding $w$ of original prompt $d$ to a distorted counterpart $\wtilde w$. $w$ and $\wtilde w$ follow distributions $P$ and $\wtilde P$, respectively. The client then converts the distorted embeddings $\wtilde w$ back to the protected prompt $\wtilde d$, which will be shared with the server for LLM inference. In this way, the chance that the attacker can infer the original prompt $d$ based on $\wtilde w$ is reduced. The distribution $\wtilde P$ is termed \textit{protected distribution} pertaining to client's prompt. 

Let us consider the process of adding randomness as an illustrative case. We posit that $w$ follows the distribution $P$, which is a normal distribution with mean $\mu_0$ and covariance matrix $\Sigma_0$, and the perturbative noise $\epsilon$ is also normally distributed with a zero mean and noise covariance $\Sigma_{\epsilon}$. Here, $\Sigma_0$ is specified as a diagonal matrix with elements $\{\sigma_{1}^2, \ldots, \sigma_{m}^2\}$, and $\Sigma_\epsilon$ is similarly a diagonal matrix with identical elements $\{\sigma_\epsilon^2, \ldots, \sigma_\epsilon^2\}$. Subsequently, original prompt is safeguarded through the addition of noise to its embedding: $\wtilde w = w + \epsilon\sim\calN(\mu_0, \Sigma_0+ \Sigma_\epsilon)$ with distribution $\wtilde P = \calN(\mu_0, \Sigma_0+ \Sigma_\epsilon)$. The randomization mechanisms can defend against a broad range of privacy attacks, for it directly replaces original tokens by sampling noise.

With the protection mechanism, the LLM inference procedure is illustrated in Figure \ref{fig:llm_infer_framework} and described as follows:
\begin{enumerate}[label=\circled{\arabic*}]
\item The client designs a prompt $d$ to accomplish a specific task.
\item The client leverages a privacy protection mechanism $\calM$ to transform $d$ to protected version $\wtilde{d}$, aiming to prevent the server (i.e., the attacker) from inferring private information through investigating $\wtilde{d}$. 
\item The client sends the prompt of $\wtilde{d}$ to the server to query the LLM.
\item The LLM at the server takes the protected prompt $\wtilde{d}$ as input and generates the corresponding response $\wtilde{r}$. Then, the server sends $\wtilde{r}$ to the client.
\end{enumerate}

\begin{figure*}[t!]
    \centering
    \includegraphics[scale=0.25]{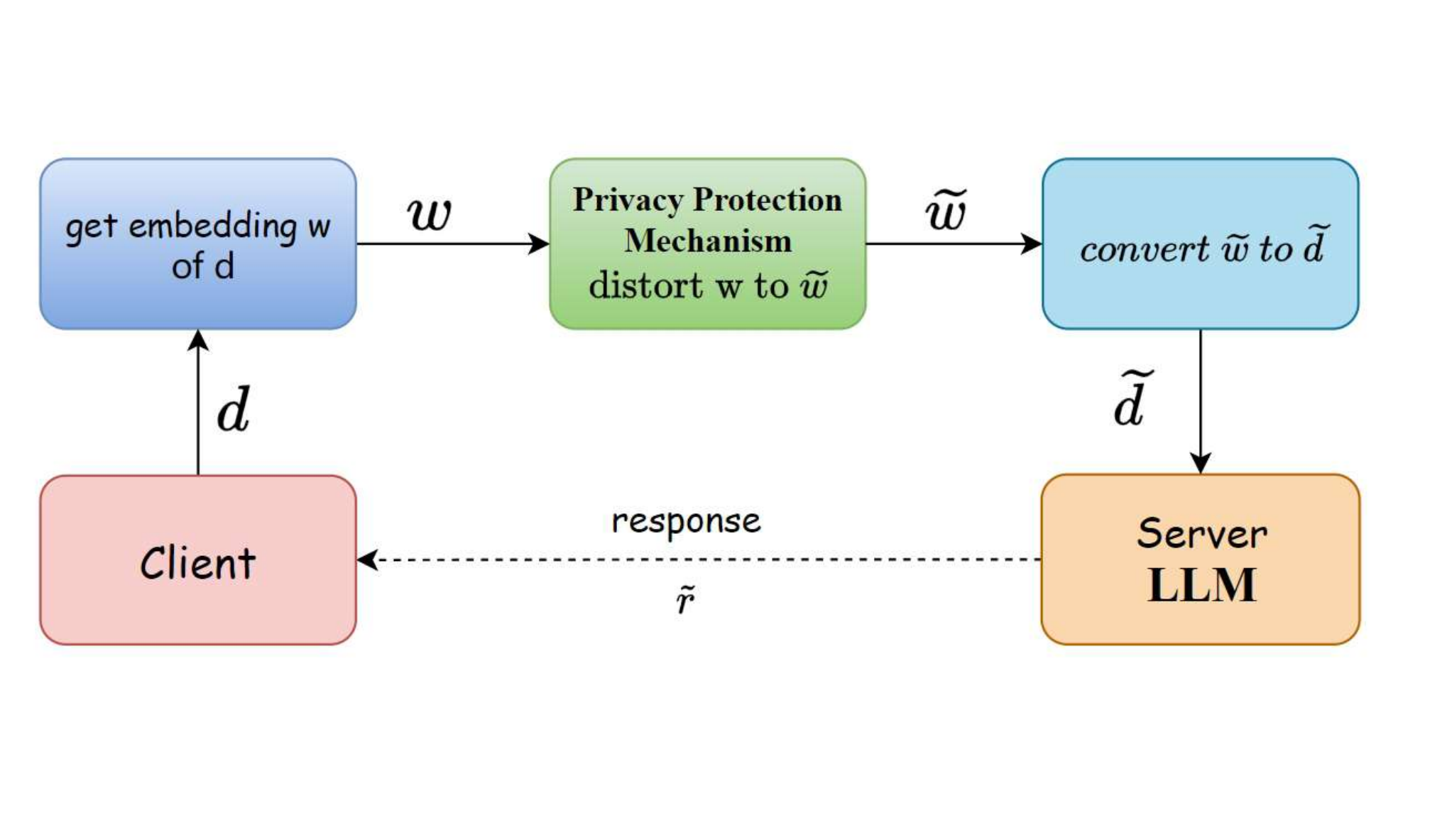}
    \caption{Illustration of the privacy-preserving LLM inference. The LLM inference procedure involves four steps: (1) the client designs a prompt $d$; (2) the client applies a privacy protection mechanism to the original prompt $d$ to obtain a protected prompt $\wtilde{d}$; (3) the client sends $\wtilde{d}$ the server; (4) The LLM takes $\wtilde{d}$ as input and sends the generated response $\wtilde{r}$ back to the client.}
    \label{fig:llm_infer_framework}
\end{figure*}

In our study, black-box LLM is mainly considered by us, whose architecture, model parameters, and inference details are hidden behind commercial APIs and user interfaces. Therefore, protection mechanisms involving encryption and secure multi-party computation do not apply to our framework. To protect the privacy of prompts querying a black-box LLM, randomization is the representative privacy protection mechanism investigated in the literature. In the next subsection, we elucidate the randomization protection mechanism.


\subsubsection{Randomization Protection Mechanism}

Given a tokenizer, a token vocabulary $V$, and $E \in$  $\mathbb{R}^{|V|\times M}$, an embedding model, where $|V|$ is vocabulary's size and $M$ is embedding's dimension, respectively, we first use tokenizer to turn prompt $d$ into tokens $\{d^{(m)}\}_{m=1}^{|V|}$, where $d^{(m)} \in$ $V$, and then we employ $E$ to map a token $d^{(m)}$ into an embedding $w^{(m)} = E(d^{(m)})$. The random distortion of a token $d^{(m)}$ can be done by the addition of a random noise $\delta$ to $w^{(m)}$: $\wtilde{w}^{(m)} = w^{(m)} + \delta$ and replace $d^{(m)}$ with a token $\wtilde{d}^{(m)}$ whose embedding is close to $\wtilde{w}^{(m)}$, the procedure of which is illustrated in Figure \ref{fig:llm_perturbation}.

\begin{figure*}[h!]
    \centering
    \includegraphics[width=0.99 \linewidth]{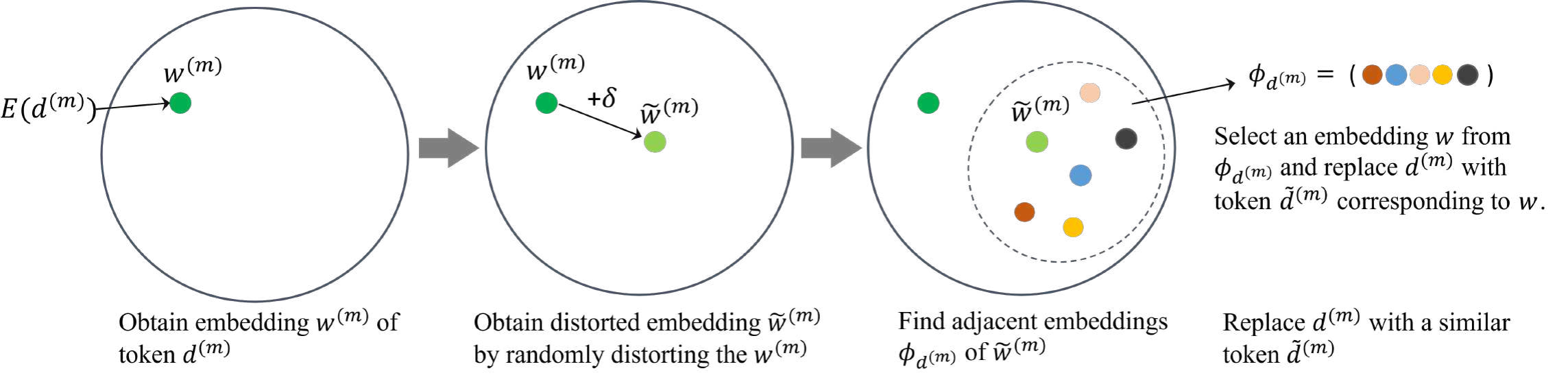}
    \caption{Illustration of randomization protection mechanism. The randomization protection mechanism aims to replace each token in a prompt with a semantically similar token. It typically involves four steps: (1) obtain the embedding $w^{(m)}$ of a given token $d^{(m)}$; (2) add random noise $\delta$ to the embedding $w^{(m)}$ to obtain a distorted embedding $\wtilde{w}^{(m)}$; (3) find adjacent embeddings of $\wtilde{w}^{(m)}$, denoted as $\phi_{d^{(m)}}$; (4) Select an embedding $w$ from $\phi_{d^{(m)}}$ and replace $d^{(m)}$ with token $\wtilde{d}^{(m)}$  corresponding to $w$.}
    \label{fig:llm_perturbation}
\end{figure*}

The literature has proposed various methods to compute $\delta$ and choose a semantic similar token $\wtilde{d}^{(m)}$ for $d^{(m)}$ based on the distorted embedding $\wtilde{w}^{(m)}$ of $w^{(m)}$. Herein, we give two examples:

\begin{itemize}
\item Applying $d_{\chi}$-privacy to perturb token $d^{(m)}$~\cite{li2023privacy}. Given dimension $\pi$ and privacy parameter $\eta$, the random noise $\delta$ is computed by $\delta = lv$,  where the scalar $l \sim$ Gamma distribution $ \Gamma(\pi, \frac{1}{\eta} )$ and the uniform vector $v$ is selected at random in accordance with a uniform distribution across the entire unit ball $\mathbb{B}^\pi$. we first compute its embedding $w^{(m)} = E(d^{(m)})$ and random noise embedding $\wtilde{w}^{(m)} = w^{(m)} + \delta$. A token $d^{(m)}$ is replaced with a semantically similar token $\wtilde{d}^{(m)}$ by solving the optimization problem: $\wtilde{d}^{(m)} = \argmin_{k}||w^{(m)} - \wtilde{w}^{(m)}||$. By repetitively replacing each token $d^{(m)}$ in $d$ with $\wtilde{d}^{(m)}$, we obtain a privatized textual version $\wtilde{d}$ of the original prompt $d$.

\item Applying the random adjacent list to perturb $d^{(m)}$~\cite{tong2023infer}. Given a token $d^{(m)}$, we first compute its embedding $w^{(m)} = E(d^{(m)})$ and random noise embedding $\wtilde{w}^{(m)} = w^{(m)} + \delta$. The random adjacency list of $d^{(m)}$ is composed of any token $d^{(k)} \in V, k \neq m$ whose embedding $w^{(k)}$ has a Euclidean distance to $w^{(m)}$ that is shorter than Euclidean distance between $w^{(m)}$ and $\wtilde{w}^{(m)}$. We denote the random adjacent list of $d^{(m)}$ as $\phi_{d^{(m)}}$. By repetitively replacing each token $d^{(m)}$ of $d$ with a token in $\phi_{d^{(m)}}$, we obtain a privatized textual version $\wtilde{d}$ of $d$.
 
\end{itemize}


We now proceed to provide formal definitions for the fundamental concepts imposed by the privacy-preserving mechanism $\mathcal{M}$ and the attacking mechanism $\mathcal{A}$.

\subsection{Utility Loss and Privacy Leakage}


The original prompt is protected by protection mechanism $\mathcal{M}$ perturbing its embedding, which affects both the loss of utility and leakage of privacy of the prompt (see Section \ref{sec:protect}), and the attacking method $\mathcal{A}$ typically reconstructs tokens of the original prompt based on the observed (i.e., protected) prompt (see Section \ref{sec:attack}), which determines the privacy leakage. In this section, we formally defined privacy leakage and utility loss.

\blue{
\begin{definition}[Privacy Leakage]\label{defi: privacy_leakage}
Let $\wtilde P$ denote protected embedding's distribution of the client's prompt, and $\breve P$ denote the distribution of embedding that is independent of the client's prompt. We define privacy leakage as the measurement that quantifies the gap between the recovery extent of the private information when the distorted prompt is provided and the random prompt is provided to the attacker:
\begin{align}
    \epsilon_p = R(\wtilde P) - R(\breve P) 
\end{align}
where $R(\wtilde P) = \mathbb E_{w\sim \wtilde P} [R(w)]$ and $R(\breve P) = \mathbb E_{w\sim \breve P}[R(w)]$. The $R(w)$ is  the recovery extent $R$ on the embedding $w$ of prompt $d$ and we define it as
\begin{align}\label{eq: defi_privacy_leakage}
   R(w) = 1 - \frac{1}{I}\sum_{i = 1}^{I} \frac{||\frac{1}{|d|}\sum_{m=1}^{|d|} (d_i^{(m)}(w) - d^{(m)})||}{\Omega}.
\end{align}
where $d^{(m)}$ represents prompt $d$'s original $m$-th token, $d_i^{(m)}(w)$ represents the $m$-th token inferred by the attacker at iteration $i$ upon observing the embedding $w$, and $I$ represents the total number of the iterative process of the attacking approach launched by the attacker. 

\noindent\textbf{Remark:}\\
(1) $R(\breve P)$ acts as the baseline scenario where the attacker attempts to recover the original prompt through random guessing. In essence, if the attacker achieves a recovery extent close to $R(\breve P)$, it indicates that the privacy leakage incurred by the attacker is approaching 0.\\
(2) We assume that $||d_i^{(m)} - d^{(m)}||\in [0,\Omega]$, where $\Omega$ is the upper bound of the distance between tokens and it is a positive constant. Therefore, $\epsilon_p\in [0,1]$.\\
(3) When all estimated tokens $d_i^{(m)}$ are equal to the original token $d^{(m)}$ for every index $m$ ranging from 1 to the length of the token sequence $|d|$, it indicates that there is no differential privacy protection applied, and the privacy leakage is at its maximum (quantified as 1).\\
(4) Without loss of generality, we posit privacy parameter $I$ is greater than 0. If $I$ is equal to 0, it implies that no privacy protection is enforced, and the privacy loss for any token $w$ will be 0.
\end{definition}

\begin{assumption}\label{assump: bound_on_R}
    We assume that $R(\wtilde w) - R(w) \ge R(\wtilde w)/c$ $\forall w\in{\mathcal W}$ and  $\wtilde w\in\wtilde {\mathcal W}$, where $c > 0$ represents a constant.
\end{assumption}
}






This following definition captures the discrepancy in utility between using unprotected model information and using model information with privacy-preserving mechanisms, and it serves as a quantitative measure of the impact of privacy protection on the utility of the model inference.

\begin{definition}[Utility Loss]\label{defi: utility_loss}
The decrement in utility reflects the contrast between the utility achieved through the model's information in the absence of privacy-preserving measures, adhering to the original distribution $P$, and the utility achieved when privacy safeguards are implemented, conforming to the secured distribution $\widetilde{P}$. Mathematically, it is defined as:

\begin{align*}
\epsilon_{u} = U(P) - U(\widetilde{P}),
\end{align*}
where $P$ and $\widetilde{P}$ represent the distributions of models without and with a privacy protection mechanism, respectively. The expected utility $U(P) = \mathbb E_{s\sim P_0}\mathbb E_{w \sim P}U(w,s)$ is calculated with respect to a test dataset $s \sim P_0$, where $P_0$ is the distribution of the test dataset. The utility function $U(w, s)$ quantifies the usefulness or effectiveness of the model $w$ on the test dataset $s$.
\end{definition}




\subsection{Privacy-Preserving LLM Inference Optimization}

During inference of privacy-protecting LLM, the client aims to find a prompt protected by a protection mechanism, with the goal of minimizing utility loss given a privacy budget. Simultaneously, the server wants to infer confidential details of the client from protected prompt. We formulate the optimization problem of the privacy-preserving LLM inference as follows. 
\begin{equation}\label{eq:llm_priv_infer}
\begin{split}
    d^{*} &= \argmin\limits_{d} \epsilon_{u}(d | \mathcal{M}) \\
\text{s.t.  } & \epsilon_{p} (d | \mathcal{M}, \mathcal{A})\le\xi.
\end{split}
\end{equation}
where $\mathcal{M}$ is the privacy protection mechanism applied by the client to protect its original prompt $d$, $\mathcal{A}$ is the attacking approach leveraged by the server to reconstruct $d$ based on the observed prompt $\wtilde{d}$ sent by the client, and $\xi$ is the constraint of leakage of privacy; $\epsilon_u$ is utility loss of $d$, affected by the $\mathcal{M}$; $\epsilon_p$ is the privacy leakage on $d$, which is affected by both $\mathcal{M}$ and $\mathcal{A}$. 

The LLM inference problem formulated in Eq.(\ref{eq:llm_priv_infer}) aims to find a prompt $d^*$ that achieves the minimal utility loss while keeping the privacy leakage under an acceptable constraint $\xi$. Within the privacy-preserving LLM inference's framework, the problem to achieve minimal privacy leakage and utility loss simultaneously raises a fundamental question: Are there any possibility of designing protective mechanism to fulfill both objectives? In the forthcoming section, we address this question by presenting a counterargument in specific scenarios. We are thus motivated to introduce the notion of no free lunch in the realm of privacy-protected large language model inference, as detailed in Theorem \ref{thm: utility-privacy trade-off_mt}. The theorem establishes the inherent equilibrium between private information and utility in LLM inference, highlighting challenges in simultaneously minimizing utility loss and privacy leakage. By examining the trade-off, valuable insights into the limitations and complexities associated with achieving optimal privacy and utility guarantees are derived in LLM inference scenarios.



\section{No Free Lunch Theorem for Privacy-Preserving LLM Inference}



The fundamental concept behind privacy-preserving LLM inference revolves around perturbing the embedding of the original prompt $d$ and transmitting the perturbed embedding $\widetilde{w}$ as the prompt $\widetilde{d}$ to the server, ensuring the privacy of $d$. In this process, the utility loss is employed to quantify the increase in the expected average loss caused by the perturbed embedding $\widetilde{w}$ compared to the utility provided by $w$. It is worth noting that $\widetilde{w}$ and $w$ follow the distributions $\widetilde{P}$ and $P$ respectively. By considering these factors, we delve into the No Free Lunch Theorem for Privacy-Preserving LLM Inference, which is the focus of this article. This theorem sheds light on the inherent equilibrium between utility and private information in LLM inference, emphasizing challenges associated with simultaneously minimizing utility loss and privacy leakage. By exploring this trade-off, we gain a deeper understanding of the intricacy of securing optimal standard of privacy with compromising least utility of LLM inference mechanisms.


In this context, our objective is to investigate the limitations of randomization-based privacy-preserving mechanisms. To accomplish our goal, it is essential to measure the unavoidable decrement in utility inherent in the process of safeguarding privacy. Logically, as we apply greater changes to the original prompt $w$, we enhance the level of privacy preservation, but at the same time, we compromise the accuracy of the resulting output. To assess the degree of distortion, we employ Euclidean distance of tokens of the original prompt and protected prompt's tokens. This distance serves as a crucial component that connects the notions of  utility loss and privacy leakage. By incorporating this distance metric, we establish a framework for understanding the equilibrium of utility loss and privacy preservation in context of randomization-based privacy-preserving LLM inference. This framework forms the basis for the formulation and exploration of the No Free Lunch Theorem for Privacy-Preserving LLM Inference, which is the central focus of our article.




\begin{definition}[Optimal Embedding]
Let $w^{*}$ denote the embedding of the prompt $d^{*}$ that maximizes the utility. Specifically, 
\begin{align*}
    w^{*} = \argmax_{w\in\mathcal W} U(w),
\end{align*}
where $U(w) = \mathbb{E}_{s\sim P_o}U(w,s)$ is the expected utility computed on a test dataset $s$ sampled from distribution $P_0$, and $\mathcal{W}$ is the union of the support $\wtilde P$ and the support of $P$.
\end{definition}

Subsequently, we define the notion of near-optimal embedding, denoted by $\mathcal{W}_c$. Here, $\wtilde{\calW}$ represents the set of possible protected embeddings, and $c$ is a non-negative constant. The near-optimal prompt embedding consists of those embeddings from $\wtilde{\calW}$ for which the difference in utility, denoted by $U(w^*)- U(w)$, is no greater than $c$ for all possible optimal embeddings $w^*$ in the set $\mathcal{W}^*$. In other words, the near-optimal embedding captures the embeddings that result in a utility loss within the tolerance specified by $c$ when compared to all possible optimal embeddings.
\begin{definition}[Near-optimal Embedding]\label{defi: neighbor_set}
Suppose $\wtilde{\calW}$ is the support of the protected embedding's distribution. Let $c$ be a constant which is no smaller than 0, we define the \textit{near-optimal prompt embedding} $\mathcal{W}_c$ as the subset of $\wtilde{\calW}$ that satisfies the following condition: for any optimal embedding $w^{*}$ in the set $\mathcal{W}^{*}$, the difference in utility between $w^{*}$ and any embedding $w$ in $\mathcal{W}_c$ is no greater than $c$, i.e.,
$$\calW_{c} = \left\{w\in\wtilde{\calW}: \left| U(w^{*})- U(w)\right|\le c, \forall w^{*}\in\mathcal W^{*}\right\}.$$
\end{definition}

In other words, near-optimal prompt embedding $\mathcal{W}_c$ consists of the embeddings from $\widetilde{\mathcal{W}}$ that exhibit a utility difference of at most $c$ compared to any optimal embedding $w^{*}$ in $\mathcal{W}^{*}$. This definition allows us to identify a set of embeddings that are close in utility to the optimal embeddings while providing a level of privacy protection.

Assumption \ref{assump: assump_of_Delta} ensures that there is a limit on the density of embeddings within a certain distance from the optimal embeddings. This prevents situations where the utility function lacks variability and fails to distinguish between optimal embeddings and a subset of embeddings.
\begin{assumption}\label{assump: assump_of_Delta}
Let $c$ in $\mathcal{W}_c$ (see Definition \ref{defi: neighbor_set}) be $\alpha$, the \textit{maximum} constant s.t.
\begin{align}
     \int_{\wtilde{\mathcal{W}}} \wtilde{p}(w)\one\{w\in\calW_{\alpha}\} dw\le\frac{{\text{TV}}(P || \wtilde{P} )}{2},
\end{align}
where $\widetilde{p}$ is distorted embedding $\widetilde{w}$'s probability density, $\alpha$ is a positive constant.
\end{assumption}
\noindent\textbf{Remark:}\\

In the above assumption, we state that $\alpha$ is the maximum constant that satisfies the inequality involving total variation distance between the true distribution $P$ and the distorted distribution $\widetilde{P}$. Such inequality gives a bound to near-optimal parameters' cumulative density, as delineated in Definition \ref{defi: neighbor_set}. We assume that $\alpha$ is positive, indicating that there exists a non-zero tolerance for the utility loss.\\
\begin{definition}[Distortion Extent]\label{defi: distortion_extent}
    Let $d^{(m)}$ and $\wtilde{d}^{(m)}$ be original prompt $d$'s and protected prompt $\wtilde{d}$'s $m$-th token, respectively. Then, the distortion extent is defined as 
    \begin{align}
        \Delta =  \|\frac{1}{N}\sum_{m = 1}^{N} g(d^{(m)}) - \frac{1}{N}\sum_{m = 1}^{N} g(\wtilde{d}^{(m)})\|,
    \end{align}
where $N = |d|$ and $||\cdot||$ is the Euclidean distance.\\
\textbf{Remark:} Assuming general applicability, we posit that $\Delta\le 1$.
\end{definition}


\begin{assumption}[Bi-Lipschitz Condition]\label{assump: two_sided_Lipschitz}
For any two prompts $d_1$ and $d_2$, it is assumed that there exist constants $c_a$ and $c_b$ s.t.:
\begin{align*}
c_a \cdot ||g(d_1) - g(d_2)|| \le ||d_1 - d_2|| \le c_b \cdot ||g(d_1) - g(d_2)||,
\end{align*}
where $g(d)$ represents the encoding of prompt $d$, and $||\cdot||$ denotes a norm (e.g., Euclidean norm).
\end{assumption}

The above assumption establishes a bi-Lipschitz condition between the distance in the prompt space and the distance in the encoded space. It implies that the encoding function $g(\cdot)$ preserves the relative distances between prompts up to a scaling factor within a certain range determined by the constants $c_a$ and $c_b$.


Based on the aforementioned context, we introduce Assumption \ref{assump: bounds_for_optimization_alg}, which assumes that the cumulative regret in the privacy-preserving LLM inference process is self-bounded. Specifically, a polynomial function of number of learning rounds $I$ bounds cumulative regret. The bound is defined as $\Theta(I^{p})$, where $p$ is a positive exponent. Difference between utility function evaluations of reconstructed data points $d_i^{(m)}$ and original data points $d^{(m)}$ is measured as regret. The constants $c_0$ and $c_2$ represent positive values that determine the bounds of the regret.


\begin{assumption}[Self-bounded Regret]\label{assump: bounds_for_optimization_alg}
Let $I$ represent the overall count of learning iterations conducted by the semi-honest adversary. We suppose cumulative regret satisfies the following condition:
\begin{align}
c_0\cdot I^{p} \le \sum_{i = 1}^{I}||g(d_i^{(m)}) - g(d^{(m)})|| \triangleq \Theta(I^{p}) \le c_2\cdot I^{p},
\end{align}
where $c_0$ and $c_2$ are positive constants, $d^{(m)}$ represents the $m$-th data point in the dataset $d$, and $d_i^{(m)}$ denotes the $m$-th recovered data point generated at iteration $i$ through the application of the optimization algorithm.
\end{assumption}

This assumption states that the cumulative regret is bounded by a polynomial function of $I$. The polynomial function is characterized by the exponent $p$, and the bounds are determined by the positive constants $c_0$ and $c_2$. The assumption implies that the attacker achieves regret that scales polynomially with $I$, indicating that the attacker exploits an asymptotically optimal regret for gradient-matching.

The provided assumption establishes a performance guarantee for the optimization algorithm used by the attacker and suggests that many classical gradient-matching optimizers satisfy this self-bounded regret condition.



\noindent\textbf{Remark:} In general, Assumption \ref{assump: bounds_for_optimization_alg} indicates that the attacker will exploit asymptotically optimal regret for gradient-matching. This assumption is reasonable in practice as many classical gradient-matching optimizers satisfy it. The following examples illustrate this point:

\noindent \textbf{Example 1:} The AdaGrad algorithm \cite{duchi2011adaptive} achieves an optimal regret bound of $\Theta(\sqrt{I})$ given by $O(\max{\log h, h^{1-\beta/2}\sqrt{I}})$, where $\beta\in (1,2)$ and $h$ is the dimension of the data. In this case, $p = 1/2$.

\noindent \textbf{Example 2:} The Adam algorithm \cite{kingma2014adam} achieves an optimal regret bound of $\Theta(\sqrt{I})$ given by $O(\log h\sqrt{I})$ with an improved constant, where $h$ is the dimension of the data. In this case, $p = 1/2$.

These examples demonstrate that well-known optimization algorithms like AdaGrad and Adam satisfy Assumption \ref{assump: bounds_for_optimization_alg} and achieve optimal regret bounds with specific values of $p$. This supports the practicality and applicability of the assumption in various scenarios.

\blue{
In the following Lemma \ref{lem: bound_for_privacy_leakage_mt_0}, we establish a theoretical relationship between privacy leakage and the extent of distortion in the embedding process. Assuming the Lipschitz continuity property and certain bounds on the optimization algorithm, and that a function of the number of rounds bounds expected regret of attacker's optimization algorithm. We propose a lower bound on the privacy leakage, indicating that the reconstruction of the original prompt from the protected prompt is limited. This result offers significant understanding of the trade-off between privacy and distortion in presence of a semi-honest attacker.

\begin{lem}[Theoretical Relationship between Recovery Extent and Distortion Extent]\label{lem: bound_for_privacy_leakage_mt_0}
Let \pref{assump: two_sided_Lipschitz} hold. The semi-honest attacker employs an optimization algorithm to conduct attack in order to infer the original prompt $d$ belonging to client from $\wtilde{d}$. Suppose embedding of $d$ is $w$. Let $d^{(m)}$ be $m$-th token of $d$, $\wtilde{d}^{(m)}$ represent $m$-th token of $\wtilde d$. Let the distortion extent of the embedding be introduced in \pref{defi: distortion_extent}. Let \pref{assump: bounds_for_optimization_alg} hold. That is, total $I$ rounds of the optimization algorithm is expected regret is $\Theta(I^p)$. We have that
\begin{align}
    R(w) \ge 1 - \frac{c_b\Delta + c_2\cdot c_b I^{p-1}}{\Omega},
\end{align}
where $c_2$ is introduced in \pref{assump: bounds_for_optimization_alg}, $c_b$ is introduced in \pref{assump: two_sided_Lipschitz}, and $\Omega$ is introduced in \pref{defi: privacy_leakage}.
\end{lem}
}


\blue{
The following Lemma \ref{lem: total_variation-utility trade-off_mt_0} establishes a theoretical relationship between privacy leakage and total variation distance in the context of prompt distortion. By considering prompt embedding's distributions before and after distortion, we derive a lower bound on leakage of the privacy measured by $\epsilon_{p}$. This is expressed with respect to these distributions' total variance distance. This result provides a theoretical interconnection between privacy and TV distance, indicating that as the total variation distance increases, the privacy leakage also increases. The derived bound involves constants related to the Lipschitz continuity property, the optimization algorithm, and the overall distortion extent. This analysis contributes to the understanding of the interplay between privacy and utility in the presence of prompt distortion.
}

\begin{lem}[Theoretical Relationship between Privacy Leakage and Total Variation Distance]\label{lem: total_variation-utility trade-off_mt_0}
Let $\epsilon_{p}$ be defined in \pref{defi: privacy_leakage}. Let $\wtilde P$ and $\breve P$ denote the distorted prompt embedding distribution and the independent random prompt embedding distribution. Given the preceding conditions, it follows
\begin{align*}
    \epsilon_{p} \ge C_1\cdot {\text{TV}}(\wtilde P || \breve P ),
\end{align*}
where $C_1 = \frac{\left(1 - \frac{c_b + c_2\cdot c_b I^{p-1}}{\Omega}\right)}{c}$, $c_2$ is introduced in \pref{assump: bounds_for_optimization_alg}, $c_b$ is introduced in \pref{assump: two_sided_Lipschitz}, and $c$ is introduced in \pref{assump: bound_on_R}. For an in-depth examination, consult \pref{sec: relation_privacy_total_variation} on the nexus between privacy and total variation.\end{lem}

\blue{
In the following Lemma \ref{lem: relation_utility_distortion_mt} establishes a theoretical relationship between utility loss and the total variation distance in the context of prompt distortion. Considering prompt embedding's distributions before and after distortion, we derive a lower bound on the utility loss measured by $\epsilon_{u}$. The bound is expressed with respect to these distributions' total variation distance. This result reveals that as the total variation distance increases, the utility loss also increases. The bound is characterized by the constant $\alpha$, which is related to the assumptions regarding the distorted embedding. This theorem provides valuable insights into understanding the impact of prompt distortion on the utility of the embedding, highlighting the trade-off between utility preservation and the level of distortion incurred.
}

\begin{lem}[Theoretical Relationship between Utility Loss and Total Variation Distance]\label{lem: relation_utility_distortion_mt}
Suppose \pref{assump: assump_of_Delta} is valid, and $\epsilon_{u}$ be defined in \pref{defi: utility_loss}. Let $P$ and $
\wtilde P$ represent embedding's distribution before and after being distorted. Given the preceding conditions, it follows
\begin{align*}
    \epsilon_{u} \ge \frac{\alpha}{2}\cdot {\text{TV}}(P || \wtilde{P}),
\end{align*}
where $\alpha$ is introduced in \pref{assump: assump_of_Delta}. For an in-depth examination, consult \pref{sec: relation_total_variation_utility} for detailed analysis.
\end{lem}




\blue{Upon combining Lemma 4.2 with Lemma 4.3, we deduce a quantitative relation between decrement of utility and leakage of privacy, as articulated in Theorem~\ref{thm: utility-privacy trade-off_mt}.}

\begin{thm}[No Free Lunch Theorem for Privacy-Preserving LLM Inference]\label{thm: utility-privacy trade-off_mt} 
Suppose $\epsilon_p$ be defined in \pref{defi: privacy_leakage}, and define $\epsilon_u$ as in \pref{defi: utility_loss}, given \pref{assump: assump_of_Delta}, it follows

\begin{align}\label{thm:nfl}
    \frac{C_2}{C_1}\cdot\epsilon_p + \epsilon_u\ge C_2\cdot{\text{TV}}(P|| \breve{P}),
\end{align}
in which
\begin{itemize}
\item $C_1 = \frac{\left(1 - \frac{c_b + c_2\cdot c_b I^{p-1}}{\Omega}\right)}{c}$, where $c_2$ is introduced in \pref{assump: bounds_for_optimization_alg}, $c_b$ is introduced in \pref{assump: two_sided_Lipschitz}, and $c$ is introduced in \pref{assump: bound_on_R}.

\item  $C_2 = \frac{\alpha}{2}$, where $\alpha$ is introduced in \pref{assump: assump_of_Delta}.

\item  $\text{TV}(P|| \breve{P})$ is a constant representing the distribution's total variation distance of the undistorted embedding and the distribution's total variation distance of embedding independent of the client's prompt embedding. This constant is independent of the protection mechanisms.
\end{itemize}

For an in-depth analysis, see \pref{sec: analysis_NFL_LLM}. 

\end{thm}


\blue{

\pref{thm: utility-privacy trade-off_mt} demonstrates that sum of utility loss and privacy leakage incurred by protecting the prompt is lower bounded by a constant that is contingent on specific problem. This theorem essentially asserts that when interacting with an LLM using protected prompts, the client cannot attain infinitesimal levels of privacy leakage and utility loss simultaneously. Instead, \textit{a trade-off must be made, wherein a reduction in privacy leakage ($\epsilon_{p}$) is accompanied by a corresponding increase in utility loss ($\epsilon_u$) and vice versa}.

This principle is similar to the No-Free-Lunch theorem for federated learning proposed by Zheng et al.~\cite{zhang2022no}, which demonstrated that the clients have to trade-off between privacy and utility when they mutually train a global model in privacy-preserving setting. Our study differs from Zhang et al.'s work in that we study the privacy-utility trade-off during LLM inference, whereby prompts and LLM responses are exchanged between a client and a server. In contrast, Zhang et al.~\cite{zhang2022no} focused on federated learning, during which the clients and the server exchange model parameters and gradients. This difference requires a fundamentally distinct definition of privacy leakage, thereby entailing a novel theoretical analysis of the privacy-utility trade-off. 

}



\section{Experiment}
In this section, we first introduce InferDPT, a privacy-preserving algorithm for LLM inference proposed by \cite{tong2023infer}, which safeguards user prompt privacy by injecting controlled noise into the embedding vectors of the original prompts, with the aim of exploring the inherent trade-off between privacy leakage and utility loss. Next, we outline the experimental setup and apply our definitions of privacy leakage and utility loss to evaluate how these factors evolve as the level of privacy protection is reduced.


\subsection{InferDPT Algorithm Overview}
The InferDPT framework is an innovative application of differential privacy\cite{dwork2006differential} in LLM inference, designed to enable privacy-preserving text generation, particularly for remote black-box LLMs. The framework comprises two main components: a perturbation module and an extraction module.

\textbf{Perturbation Module: }
The perturbation module protects privacy by introducing noise to the embedding vectors of each token in the original document, leveraging a differential privacy mechanism. The noise level is controlled by the privacy budget, $\epsilon$, with smaller values of $\epsilon$ providing stronger privacy guarantees. Additionally, the module generates a random adjacency list for each token, selecting replacement tokens to further strengthen privacy protection. This approach maintains the semantic coherence of the text while ensuring user privacy. The final perturbed document is created by applying these steps to each token in the original text. Subsequently, the user combines the perturbed document with task instructions to form a complete prompt, which is submitted to the remote black-box LLM. The LLM then generates a privacy-preserving output based on the perturbed input.

\textbf{Extraction Module: }
The extraction module employs a local LLM, which is lightweight and deployable, while less capable than the remote black-box LLM. This local LLM generates text based on the original document and the privacy-preserving output as a reference.

The InferDPT algorithm offers several key advantages, including effective user prompt privacy protection, resistance to embedding inversion attacks, and its high-quality text generation. Furthermore, the algorithm is highly adaptable to various local LLM and deployment environments. Its modular design—comprising a perturbation module and an extraction module—enhances scalability, accommodating diverse privacy requirements and application scenarios. Specifically, InferDPT enables differential privacy across a wide range of text data scales and complexities. InferDPT incurs lower computational overhead compared to homomorphic encryption approaches (e.g., CipherGPT\cite{hou2023ciphergpt}), effectively balancing computational costs through dynamic sampling. Experimental results show that InferDPT maintains text generation quality on par with non-private GPT-4 across multiple datasets and a range of \(\epsilon\) values.

However, the InferDPT algorithm also has some limitations. The quality of the generated text depends on the choice of the local LLM and the design of prompt. More importantly, it requires a trade-off between privacy protection and text generation quality, especially when strong privacy guarantees are applied. Additionally, deploying InferDPT requires certain hardware resources to run the local LLM, which may not be user-friendly for ordinary users.

Based on the above analysis, we believe that InferDPT is highly suitable for validating our theory. We will build upon this algorithm to explore the trade-off between privacy leakage and utility loss.

\subsection{Experiment Setup}
\textbf{Dataset:} We used the CNN/Daily Mail dataset\cite{hermann2015teachingmachinesreadcomprehend} to perform an open-ended text generation task. Following previous works\cite{tong2023infer}\cite{welleck2019neuraltextgenerationunlikelihood}\cite{xu2022learningbreakloopanalyzing}, perturbed document containing 50 tokens is provided to the remote black-box LLM or the local LLM, with additional task descriptions appended to form a complete prompt (all prompts can be found in \ref{prompt}). LLM then generates the following 100 tokens based on this prompt.

\textbf{Model:} We used GPT-3.5-turbo as the remote black-box LLM and Vicuna-7b-4bit as the local LLM. Both models were configured with a temperature setting of $0.5$ and a maximum token generation limit (\textit{max\_tokens}) of $150$.

\textbf{Quantifying Privacy Leakage $\epsilon_p$:} 
We first randomly select 50 tokens from the vocabulary to simulate the remote black-box LLM's attempt to recover the perturbed document through random guessing. However, we find that directly calculating the recovery extent $R$ as defined in Definition \ref{defi: privacy_leakage} is problematic because the black-box LLM's responses do not have a direct correspondence with the perturbation tokens. To address this issue, we use cosine similarity as a measure of the recovery extent. Cosine similarity is defined as $cos(\theta) = \frac{A \cdot B}{|A| |B|}$, which is consistent with the definition of recovery extent in Definition \ref{defi: privacy_leakage}, where $cos(\theta)$ closer to 0 indicates a higher recovery extent, and $cos(\theta)$ closer to 1 indicates a lower recovery extent.

We calculate the cosine similarity $cos(\theta))_1$ between the randomly guessed document and the original document, denoted as $R(\breve P)$, which represents the recovery extent under random guessing. Then, we send the perturbed document to the remote black-box LLM, which attempts to recover the original document, and calculate the cosine similarity $cos(\theta))_2$ between the recovered document and the original document, denoted as $R(\wtilde P)$, which represents the recovery extent of the remote black-box LLM. According to Definition \ref{defi: privacy_leakage}, we can obtain $\epsilon_p$. Our expected result is that as the privacy budget $\epsilon$ increases, $R(\breve P)$ remains stable, while $R(\wtilde P)$ gradually increases, indicating that the remote black-box LLM's ability to recover the user's privacy increases, leading to an increase in privacy leakage $\epsilon_p$.

\textbf{Quantifying Utility Loss $\epsilon_u$:} According to our definition of utility loss in \ref{defi: utility_loss}, Utility loss quantifies the difference in performance between LLM inference with and without privacy protection. Specifically, we first send the original document to the remote black-box LLM for a text continuation task, obtaining 100 generated tokens. We calculate the utility of this output as \( U(P) \). Next, we perturb the original document to create a protected version and send it to the remote black-box LLM again, generating another set of 100 tokens. These 100 tokens, along with the original prompt, are then fed into the local LLM to extract information and calculate its utility metrics, denoted as \( U(\widetilde{P}) \). The utility loss of InferDPT is quantified as the difference between these two utilities.

To evaluate utility, we employ several metrics commonly used in open-ended text generation tasks, including BERTScore \cite{zhang2020bertscoreevaluatingtextgeneration}, BLEU, Keyword Coverage, Semantic Similarity, Diversity, Coherence, and ROUGE (ROUGE-1, ROUGE-2, and ROUGE-L). 
 These metrics are described in detail as follows:

\begin{itemize}
    \item \textbf{BERTScore}: Measures the semantic similarity between the generated and reference texts using BERT embeddings, reflecting content consistency. It is particularly effective for evaluating semantic fidelity, as it captures the contextual meaning of words through pre-trained BERT models.
    \item \textbf{BLEU}: Evaluates the precision of n-grams in the generated text compared to the reference text, providing a measure of fluency and faithfulness. BLEU is widely used in tasks such as machine translation and text generation.
    \item \textbf{Keyword Coverage}: Assesses the extent to which the generated text retains key concepts or terms from the reference text. This ensures that essential information or keywords are accurately represented in the output.
    \item \textbf{Semantic Similarity}: Beyond BERTScore, this metric captures deeper contextual and conceptual alignment between the generated and reference texts, providing a complementary perspective on semantic fidelity.
    \item \textbf{Diversity}: Measures the variety and richness of the generated text, ensuring that the model avoids repetitive outputs. This is particularly important under privacy constraints, which can sometimes reduce variability in the content.
    \item \textbf{Coherence}: Evaluates the logical flow and structural consistency of the generated text, ensuring that the output aligns with the expected narrative or topic and maintains readability.
    \item \textbf{ROUGE Scores (ROUGE-1, ROUGE-2, and ROUGE-L)}: Measure the overlap of n-grams between the generated and reference texts. ROUGE-1 focuses on unigram overlap, ROUGE-2 on bigrams, and ROUGE-L considers the longest common subsequences. These metrics are widely used for assessing textual accuracy, coverage, and relevance.
\end{itemize}

We divided privacy protection into 24 distinct levels and calculated both privacy leakage and utility loss for each level. By plotting privacy leakage on the x-axis and utility loss on the y-axis, we mapped points corresponding to different privacy budgets or noise intensities. As expected, utility loss decreases gradually as privacy leakage increases, indicating a reduction in privacy protection strength.

\subsection{Design of Prompts for the InferDPT Framework}\label{prompt}
In our study, we employed different quantification strategies for measuring privacy leakage and utility loss, which required distinct prompt designs. Specifically, for quantifying privacy leakage, we tasked the remote black-box LLM with attempting to recover the perturbed documents. For quantifying utility loss, we instructed the remote black-box LLM to complete continuation tasks based on the input text, and then used the local LLM to extract relevant content from the original document.

\subsubsection{Prompts for Quantifying Privacy Leakage}

We designed the following prompt to instruct the remote black-box LLM to act as an attacker attempting to recover the perturbed documents:

\textit{You are a text recovery expert. Your task is to recover the original tokens that were replaced by their semantically similar alternatives.}

\textit{Examples:}
\begin{itemize}
    \item \textit{Original: The company announced record profits in the third quarter.}\\
          \textit{Perturbed: Business speak happy money three time part.}\\
          \textit{Recovered: The business reported positive financial results in the third period.}
    \item \textit{Original: Scientists discover new species in Amazon rainforest.}\\
          \textit{Perturbed: People find thing green tree place.}\\
          \textit{Recovered: Researchers identified a new organism in the tropical forest.}
\end{itemize}

\textit{Instructions:}
\begin{itemize}
    \item \textit{Perturbed:}
    \item \textit{You should only respond with the recovered tokens.}
    \item \textit{Output format: [token1] [token2] [token3] ...}
\end{itemize}

\subsubsection{Prompts for Quantifying Utility Loss}
To quantify utility loss, we first provide the original document to the remote black-box LLM to generate a continuation. Next, the remote black-box LLM is tasked with generating a continuation for the perturbed document. To evaluate the utility loss, we design prompts for the local LLM to extract the continuation content from the remote black-box LLM's responses, using the original document as a reference. Accordingly, this section presents the prompts designed for both the remote black-box LLM's continuation task and the local LLM's extraction task.

\textbf{Prompt for Remote Black-box LLM:}

\textit{
You are a creative writer skilled in continuing stories and texts in a natural and engaging way.Your task is to generate a coherent continuation that matches the style and context of the original text.}

\textit{
Continue this text in a natural and engaging way, maintaining the same style and tone. Write 2-3 sentences that flow naturally from the original text.}

\textit{
——Original Text: }

\textit{
——[your continuation]}

\textbf{Prompt for Local LLM:}

\textit{
Analyze the coherence between the "Prefix Text" and "Perturbed Generation", then follow these rules:}

\textit{
        1. If the "Perturbed Generation" meets ALL these criteria:}

\textit{
           - Continues the story/topic from the "Prefix Text" naturally}

\textit{
           - Maintains consistent facts and details}

\textit{
           - Uses similar tone and style}

\textit{
           → Output the "Perturbed Generation" exactly as is}

\textit{
        2. Otherwise, if ANY criteria are not met:}

\textit{
           - Write a new continuation that:}

\textit{
           - Directly follows from the last sentence of "Prefix Text"}

\textit{
           - Matches the style and tone of "Perturbed Generation"}

\textit{
           - Maintains factual consistency with "Prefix Text"}

\textit{
        OUTPUT ONLY the continuation text, with no explanations or additional text.}

\textit{
        ——"Prefix Text": }

\textit{
        ——"Perturbed Generation": }

\textit{
        ——"Your Continuation":}

\subsection{Results}

\begin{figure}[thpb]
  \centering
  \includegraphics[width=0.5\linewidth]{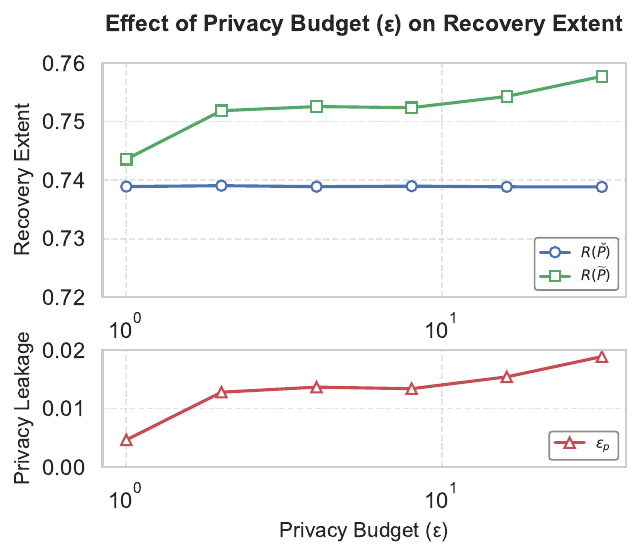}

  \caption{  \textbf{Privacy Leakage Analysis with Varying Privacy Budget ($\epsilon$).} The figure illustrates the relationship between the privacy budget \(\epsilon\) and two key metrics: recovery extent and privacy leakage. The top plot displays the recovery extents of the random document $R(\breve P)$ and the perturbed document $R(\wtilde P)$, while the bottom plot shows the privacy leakage, both as functions of the privacy budget \(\epsilon\).} 
  \label{privacy_analysis}
\end{figure}
\begin{figure}[thpb]
  \centering
  \includegraphics[width=1\linewidth]{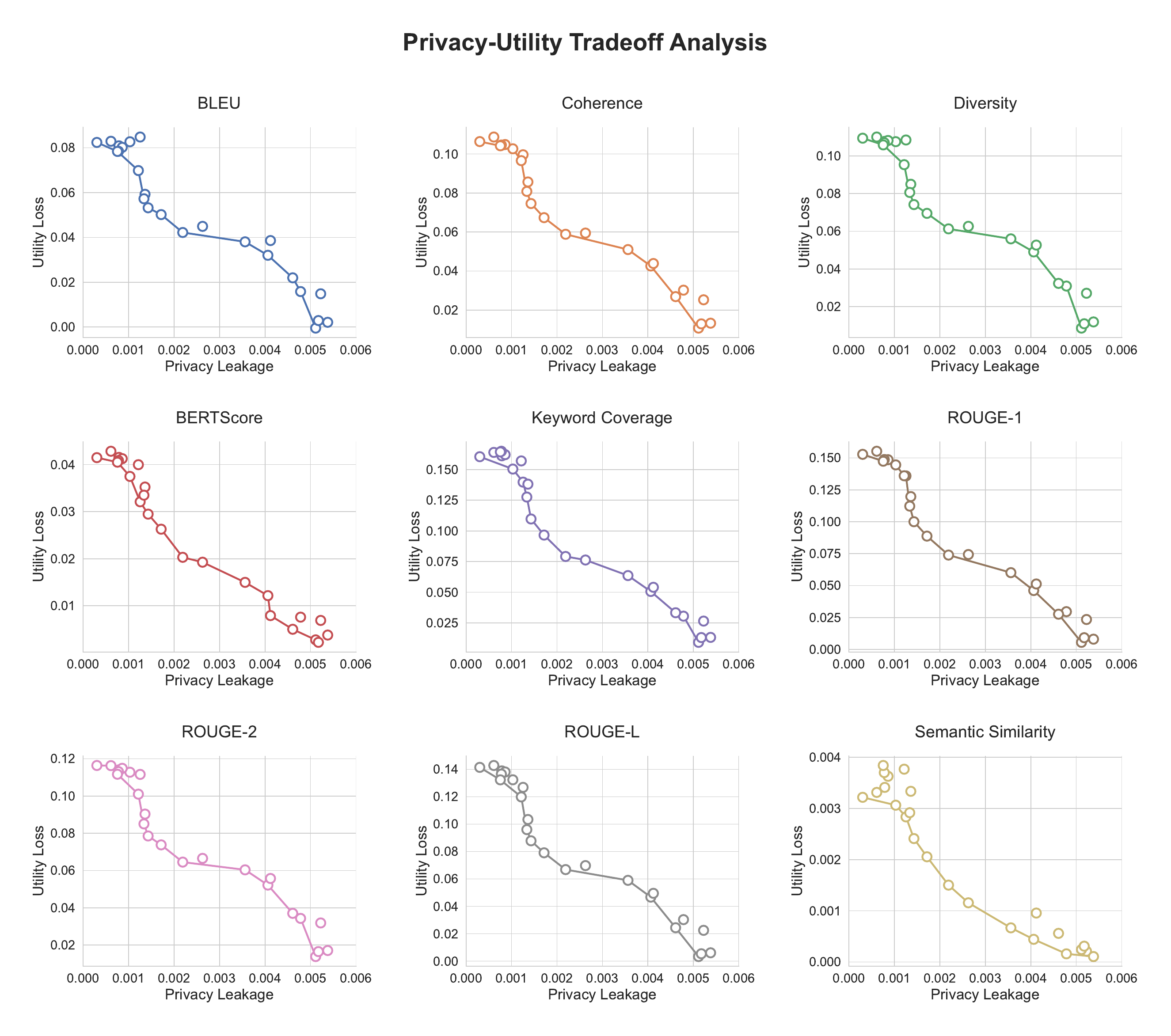}
  \caption{ \textbf{Privacy and Utility Tradeoff.} This figure illustrates the relationship between nine utility loss metrics (vertical axis) and one privacy leakage metric (horizontal axis). Each subplot shows the results for 24 different privacy budgets $\epsilon$, reflecting the corresponding utility loss and privacy leakage. As $\epsilon$ decreases, privacy leakage increases while utility loss decreases. This clearly demonstrates the tradeoff between privacy leakage and utility loss. } 
  \label{tradeoff}
\end{figure}

According to our definition of privacy leakage Definition\ref{defi: privacy_leakage}, a recovery extent $R$ closer to 0 indicates a higher amount of recovered private information. As shown in figure \ref{privacy_analysis}, $R(\wtilde P)$ increases as the privacy budget $\epsilon$ increases, suggesting that a higher $\epsilon$ reduces the difficulty of recovering protected prompts, allowing the remote black-box LLM to retrieve more private information from them. In contrast, the $R(\breve P)$ rate remains relatively stable across different $\epsilon$ values, indicating that random guessing does not affect the amount of recovered private information with changes in the privacy budget. This stability serves as a baseline, highlighting the sensitivity of protected prompt recovery to changes in the privacy budget $\epsilon$ in remote black-box LLMs.

Additionally, privacy leakage ($\epsilon_p$) increases as $\epsilon$ rises, consistent with the expectation that a higher privacy budget allows for greater information leakage. At lower values of $\epsilon$, privacy leakage is minimal, and in some cases even negative, indicating strong privacy protection, where the remote black-box LLM's ability to recover protected prompts is worse than random guessing. However, as $\epsilon$ approaches higher values, privacy leakage increases, suggesting that with a larger privacy budget, the remote black-box LLM is able to recover more private information from the protected prompts, thereby increasing the privacy risk.

Although this result is consistent with our defined relationship between privacy leakage and $\epsilon$, it is important to note that we did not use a specialized recovery algorithm for the experiment. Instead, recovery was performed by sending instructions to the remote black-box LLM. We believe that employing a dedicated iterative recovery algorithm could further improve the accuracy of privacy leakage detection.

Moreover, as shown in Figure \ref{tradeoff}, this experimental result illustrates the relationship between privacy leakage and utility loss, providing a clear visualization of the trade-off between the two. The figure reveals that as privacy leakage increases—indicating a decrease in the level of privacy protection—all utility metrics (e.g., BERTScore and semantic similarity) exhibit a declining trend. This demonstrates that InferDPT requires a balance between privacy leakage and utility loss: while higher levels of privacy protection effectively reduce privacy leakage, they come at the cost of significant utility loss. Conversely, lower levels of privacy protection improve utility but compromise privacy. Only at an optimal level of privacy protection can both privacy and utility be reasonably balanced. These findings strongly validate our theory that there is an inherent trade-off between privacy leakage and utility loss, aligning well with the "no free lunch" theorem. This discovery offers valuable insights for optimizing privacy-preserving strategies in LLM inference.
\section{Conclusion}

While LLMs offer tremendous benefits in terms of productivity enhancement and data analysis, their usage also raises concerns about the security and privacy of user queries. Protecting privacy is, therefore, crucial when prompting LLMs for inference, and minimizing privacy leakage is a fundamental requirement. Various randomization approaches have been suggested for safeguarding prompts' privacy. However, adoption of protection mechanisms may introduce utility loss. This conflict between minimizing privacy leakage and utility loss motivated us to investigate whether developing a protection mechanism that simultaneously achieves utility loss and minimal privacy leakage is plausible.

To answer this question, we propose a privacy-preserving LLM inference framework. Within this framework, we propose and experimentally verified a No-Free-Lunch Theorem for privacy-preserving LLM inference, demonstrating that the weighted summation of these factors exceeds a constant which is contingent on specific problem and non-zero. This highlights inevitable loss of utility when leakage budget of privacy is excessively constrained.


Moving forward, it is essential to continue exploring novel approaches and solutions to strike an optimal balance between privacy and utility.  Our future study may concentrate on refining and expanding privacy-preserving LLM inference framework, investigating additional protection mechanisms (e.g., encryption), and evaluating their practical and theoretical effectiveness. 

\section{ACKNOWLEDGMENTS}
We thank Yulin Fei for helpful comments. This work was supported by the National Science and Technology Major Project under Grant 2022ZD0115301.

\bibliography{references, llm}
\bibliographystyle{elsarticle-num}

\newpage
\onecolumn
\appendix
\section{Analysis for \pref{lem: bound_for_privacy_leakage_mt_0}}\label{sec: privacy_distortion_extent}
\begin{lem}[Theoretical Relationship between Recovery Extent and Distortion Extent]\label{lem: bound_for_privacy_leakage_app}
Let \pref{assump: two_sided_Lipschitz} hold. The semi-honest attacker uses an optimization algorithm to infer the original prompt $d$ of the client based on the protected prompt $\wtilde{d}$. Let $w$ and $\wtilde w$ represent the embeddings of $d$ and $\wtilde{d}$, respectively. Let $d^{(m)}$ represent $m$-th token of $d$, $\wtilde{d}^{(m)}$ represent $m$-th token of $\wtilde d$. Let $\Delta = \|\frac{1}{|d|}\sum_{m = 1}^{|d|} g(d^{(m)}) - \frac{1}{|d|}\sum_{m = 1}^{|d|} g(\wtilde{d}^{(m)})\|$ represent the distortion of the embedding. Let \pref{assump: bounds_for_optimization_alg} hold. That is, the expected regret of the optimization algorithm in a total of $I$, $ I > 0$, rounds is $\Theta(I^p)$. We have that
\begin{align}
    R(w) \ge 1 - \frac{c_b\Delta + c_2\cdot c_b I^{p-1}}{\Omega},
\end{align}
where $c_2$ is introduced in \pref{assump: bounds_for_optimization_alg}, and $c_b$ is introduced in \pref{assump: two_sided_Lipschitz}.
\end{lem}

\begin{proof}
Recall that 

\begin{align}
   R(\wtilde w) = 1 - \frac{1}{I}\sum_{i = 1}^{I} \frac{||\frac{1}{|d|}\sum_{m=1}^{|d|} (d_i^{(m)}(\wtilde w) - d^{(m)})||}{\Omega}.
\end{align}

To protect privacy, the client selects a protection mechanism, which maps the original parameter $w$ to a protected parameter $\wtilde w$. After observing the protected parameter, a semi-honest adversary infers the private information using the optimization approaches. Let $d_i^{(m)}$ represent the reconstructed data at iteration $i$ using the optimization algorithm. Therefore the cumulative regret over $I$ rounds
\begin{align*}
    \sum_{i = 1}^{I} [||g(d_i^{(m)}(\wtilde w)) - w|| - ||g(d) - w||]
    & = \sum_{i = 1}^{I} [||g(d_i^{(m)}(\wtilde w)) - g(d)||]\\
    & = \Theta(I^p).
\end{align*}

Therefore, we have
\begin{align}\label{eq: regret_bounds}
   c_0\cdot I^p \le \sum_{i = 1}^{I}||g(d_i^{(m)}(\wtilde w)) - g(d)|| = \Theta(I^p) \le c_2\cdot I^p,
\end{align}
where $c_0$ and $c_2$ are constants independent of $I$.

Let $x$ and $x'$ represent two data. From our assumption, we have that
\begin{align}
    c_a ||g(x) - g(x')||\le ||x - x'||\le c_b ||g(x) - g(x')||.
\end{align}


Let $d_i^{(m)}(\wtilde w)$ represent the reconstructed $m$-th data at iteration $i$ using the optimization algorithm. We have that
\begin{align*}
    ||\frac{1}{|d|}\sum_{m = 1}^{|d|}(d_i^{(m)}(\wtilde w) - d^{(m)})||
    &\le ||\frac{1}{|d|}\sum_{m = 1}^{|d|}(\wtilde d^{(m)} - d^{(m)})|| + ||\frac{1}{|d|}\sum_{m = 1}^{|d|}(d_i^{(m)}(\wtilde w) - \wtilde d^{(m)})||\\
    & \le c_b\cdot||\frac{1}{|d|}\sum_{m = 1}^{|d|}(g(\wtilde d^{(m)}) - g(d^{(m)}))|| + c_b ||\frac{1}{|d|}\sum_{m = 1}^{|d|}(g(d_i^{(m)}(\wtilde w)) - g(\wtilde d^{(m)}))||\\
\end{align*}
where the second inequality is due to $||\frac{1}{|d|}\sum_{m = 1}^{|d|}(\wtilde d^{(m)} - d^{(m)})||\le c_b\cdot||\frac{1}{|d|}\sum_{m = 1}^{|d|}(g(\wtilde d^{(m)}) - g(d^{(m)}))||$ and $||\frac{1}{|d|}\sum_{m = 1}^{|d|}(d_i^{(m)} - \wtilde d^{(m)})||\le  c_b ||\frac{1}{|d|}\sum_{m = 1}^{|d|}(g(d_i^{(m)}(\wtilde w)) - g(\wtilde d^{(m)}))||$.


From the definition of distortion extent, we know that
\begin{align}
        \|\frac{1}{|d|}\sum_{m = 1}^{|d|} g(d^{(m)}) - \frac{1}{|d|}\sum_{m = 1}^{|d|} g(\wtilde{d}^{(m)})\| = \Delta.
    \end{align}

Therefore, we have
\begin{align*}
    \Omega(1-R(\wtilde w)) = \frac{1}{I}\sum_{i = 1}^{I} ||\frac{1}{|d|}\sum_{m = 1}^{|d|}(d_i^{(m)} - d)||
    &\le  c_b\Delta + c_b\cdot\frac{1}{I}\sum_{i = 1}^{I}||\frac{1}{|d|}\sum_{m = 1}^{|d|}(g(d_i^{(m)}(\wtilde w)) - g(\wtilde d^{(m)}))||\\
    &\le c_b\Delta + c_2\cdot c_b I^{p-1}.
\end{align*}

Note that $c_b + c_b c_2\le \Omega$. Therefore, we have that
\begin{align}
   R(\wtilde w) \ge 1 - \frac{c_b\Delta + c_2\cdot c_b I^{p-1}}{\Omega}.
\end{align}
\end{proof}

\section{\pref{lem: total_variation-utility trade-off_mt_0} Quantitative Relationship between ${\text{TV}}({\wtilde P} || \breve P )$ and Privacy Leakage}\label{sec: relation_privacy_total_variation}

\begin{lem}[Theoretical Relationship between Privacy Leakage and Total Variation Distance]\label{lem: total_variation-utility trade-off_app}
Let $\epsilon_{p}$ be defined in \pref{defi: privacy_leakage}. Let $\wtilde P$ and $\breve P$ denote the distorted prompt embedding distribution and the independent random prompt embedding distribution. Then we have,
\begin{align*}
    \epsilon_{p} \ge C_1\cdot {\text{TV}}({\wtilde P} || \breve P ),
\end{align*}
where $C_1 = \frac{\left(1 - \frac{c_b + c_2\cdot c_b I^{p-1}}{\Omega}\right)}{c}$.
\end{lem}


\begin{proof}

Let $\mathcal U = \{w\in\mathcal W_p: d\breve P(w) - d{\wtilde P}(w)> 0\}$, and $\mathcal V = \{w\in\mathcal W_p: d\breve P(w) - d{\wtilde P}(w)< 0\}$, where $\mathcal W_p$ represents the union of the supports of $\breve P$ and ${\wtilde P}$. 

For any $w\in\mathcal V$, the definition of $\mathcal V$ implies that $d{\wtilde P}(w) > d\breve P(w)\ge 0$. Therefore, $w$ belongs to the support of $\wtilde {P}$, which is denoted as $\wtilde {\mathcal{W}}$. Therefore we have that
\begin{align}\label{eq: subset_relationship_1}
    \mathcal V\subset\wtilde {\mathcal W}.
\end{align}

Similarly, we have that
\begin{align}\label{eq: subset_relationship_n}
    \mathcal U\subset\breve{\mathcal W}.
\end{align}


Recall that $w$ represents the embedding, $\breve P$ represents the distribution of the embedding after being protected, and $\breve P(w)$ represents the corresponding probability density function.

We define 
\begin{align}
    R(w) = 1 - \frac{1}{I}\sum_{i = 1}^{I} \frac{||d_i(w) - \breve d||}{\Omega}.
\end{align}

The privacy leakage is defined as
\begin{align}
    \epsilon_p = R(\wtilde P) - R(\breve P),
\end{align}
where $R(\wtilde P) = \mathbb E_{w\sim \wtilde P} [R(w)]$ and $R(\breve P) = \mathbb E_{w\sim \breve P}[R(w)]$.

Then we have that

\begin{align*}
\epsilon_{p} &= R({\wtilde P}) - R(\breve P)\\  
    &=\mathbb E_{w\sim {\wtilde P}}[R(w)] - \mathbb E_{w\sim \breve P}[R(w)]\\
     &=\int_{\mathcal W} R(w)d{\wtilde P}(w) - \int_{\mathcal W} R(w) d\breve P(w)\\
     &=\int_{\mathcal{V}} R(w)[d {\wtilde P}(w) - d \breve P(w)] - \int_{\mathcal{U}} R(w)[d \breve P(w) - d {\wtilde P}(w)]\\
     &\ge \inf_{w\in \mathcal V} R(w)\int_{\mathcal V}[d {\wtilde P}(w) - d \breve P(w)] - \sup_{w\in\mathcal{U}} R(w)\int_{\mathcal U}[d \breve P(w) - d {\wtilde P}(w)]\\
     &= \left(\inf_{w\in \mathcal V}R(w) - \sup_{w\in\mathcal{U}} R(w)\right) \int_{\mathcal V}[d {\wtilde P}(w) - d \breve P(w)].
\end{align*}

From the definition of total variation distance, we have
\begin{align}\label{eq: bound_1_term_2_JS_ratio}
    \int_\U [d {\wtilde P}(w) - d \breve P(w)] = ||{\wtilde P} - \breve P||_{\text{TV}}.
\end{align}


From \pref{lem: bound_for_privacy_leakage_app}, we know that,
\begin{align}\label{eq: lower_bound_for_R_w}
   R(w) \ge 1 - \frac{c_b\Delta + c_2\cdot c_b I^{p-1}}{\Omega}.
\end{align}

Therefore, we have that
\begin{align*}
\epsilon_{p} &\ge\left(\inf_{w\in \mathcal V}R(w) - \sup_{w\in\mathcal{U}} R(w)\right) \int_{\mathcal V}[d {\wtilde P}(w) - d \breve P(w)]\\
&\ge\inf_{w\in \mathcal V}\frac{R(w)}{c} \int_{\mathcal V}[d {\wtilde P}(w) - d \breve P(w)]\\
&\ge\inf_{w\in \mathcal V} \frac{1}{c}\left(1 - \frac{c_b\Delta + c_2\cdot c_b I^{p-1}}{\Omega}\right)\cdot {\text{TV}}({\wtilde P} || \breve P )\\
&\ge C_1\cdot {\text{TV}}({\wtilde P} || \breve P ),
\end{align*}
where the second inequality is due to \pref{assump: bound_on_R}, the third inequality is due to \pref{eq: lower_bound_for_R_w}, and the fourth inequality is due to $C_1 = \frac{\left(1 - \frac{c_b + c_2\cdot c_b I^{p-1}}{\Omega}\right)}{c}$ and $\Delta\le 1$.
\end{proof}
\section{Analysis for\pref{lem: relation_utility_distortion_mt}: Quantitative Relationship between ${\text{TV}}(P || \wtilde P )$ and $\epsilon_u$}\label{sec: relation_total_variation_utility}

\begin{lem}[Theoretical Relationship between Utility Loss and Total Variation Distance]\label{lem: relation_utility_distortionpp}

Let \pref{assump: assump_of_Delta} hold, and $\epsilon_{u}$ be defined in \pref{defi: utility_loss}. Let $P$ and $\wtilde P$ represent the distribution of the aggregated parameter before and after being protected. Then we have,
\begin{align*}
    \epsilon_{u} \ge\frac{\Delta}{2}\cdot {\text{TV}}(P || \wtilde P ).
\end{align*}
\end{lem}


\begin{proof}

Let $\mathcal U = \{w\in\mathcal W_u: d\wtilde P(w) - dP(w)> 0\}$, and $\mathcal V = \{w\in\mathcal W_u: d\wtilde P(w) - dP(w)< 0\}$, where $\mathcal W_u$ represents the union of the supports of $\wtilde P$ and $P$. 

For any $w\in\mathcal V$, the definition of $\mathcal V$ implies that $dP(w) > d\wtilde P(w)\ge 0$. Therefore, $w$ belongs to the support of $P$, which is denoted as $\mathcal W$. Therefore we have that
\begin{align}\label{eq: subset_relationship_1_utility}
    \mathcal V\subset\mathcal W_u.
\end{align}

Similarly, we have that
\begin{align}\label{eq: subset_relationship_n_utility}
    \mathcal U\subset\wtilde{\mathcal W}.
\end{align}


It is assumed that the utility of the unprotected model information achieves the maximal value at the convergence step. Therefore, we have that
\begin{align}\label{eq: subset_relationship_2_utility}
    \mathcal W \subset \mathcal W^{*}.
\end{align}

Notice that from the definition of $\mathcal W^{*}$, for any $w\in\mathcal W$ and $w^*\in\mathcal W^*$ we have that
\begin{align}\label{eq: w_best}
      U(w^{*})\ge  U(w).
\end{align}

Let $\Delta$ be a positive constant defined in \pref{assump: assump_of_Delta}, from \pref{defi: neighbor_set} we have
$$\calW_{\Delta} = \left\{w\in\wtilde{\mathcal W}: \left|  U(w^{*})- U(w)\right|\le\Delta, \forall w^{*}\in\mathcal W^{*}\right\},$$ which implies that for any $w\in\wtilde{\mathcal W}\setminus\calW_{\Delta}$ and $w^*\in\mathcal W^*$ it holds that
\begin{align}\label{eq:l_U_bound_11}
  \left|  U(w^{*})- U(w)\right|>\Delta.   
\end{align}

Combining \pref{eq: w_best} and \pref{eq:l_U_bound_11}, for any $w\in\wtilde{\mathcal W}\setminus\calW_{\Delta}$ and $w^*\in\mathcal W^*$ we have 
\begin{align}\label{eq:l_U_bound_1}
  U(w^{*})- U(w)>\Delta.
\end{align}

Recall that $\wtilde P$ represents the distribution of the aggregated parameter after being protected. Then we have
\begin{align*}
\epsilon_{u} &= U(P) - U(\wtilde P)\\  
    &=\mathbb E_{w\sim P}[U(w)] - \mathbb E_{w\sim \wtilde P}[U(w)]\\
     &=\int_{\mathcal W} U(w)dP(w) - \int_{\mathcal W} U(w) d\wtilde P(w)\\
     &=\int_{\mathcal{V}} U(w)[d P(w) - d \wtilde P(w)] - \int_{\mathcal{U}} U(w)[d \wtilde P(w) - d P(w)]\\
    &\overset{\spadesuit}{=}\int_{\mathcal V} U(w)\one\{w\in\mathcal W^{*}\}[d P(w) - d \wtilde P(w)] - \int_{\mathcal{U}} U(w)[d \wtilde P(w) - d P(w)]\\
     &\overset{\bigstar}{=}\int_{\mathcal V} U(w)\one\{w\in\mathcal W^{*}\}[d P(w) - d \wtilde P(w)] - \int_{\mathcal{U}} U(w)\one\{w\in\wtilde{\mathcal W}\}[d \wtilde P(w) - d P(w)]
\end{align*}
where $\spadesuit$ is due to $\mathcal V\subset\mathcal W_u\subset \mathcal W^{*}$ from \pref{eq: subset_relationship_1_utility} and \pref{eq: subset_relationship_2_utility}, and $\bigstar$ is due to $\mathcal U\subset\wtilde{\mathcal W}$ from \pref{eq: subset_relationship_n_utility}.

We decompose $\int_{\mathcal{U}} U(w)\one\{w\in\wtilde{\mathcal W}\}[d \wtilde P(w) - d P(w)]$ as the summation of $\int_{\mathcal{U}} U(w)\one\{w\in\wtilde{\mathcal W}\}\one\{w\in\calW_{\Delta}\}[d \wtilde P(w) - d P(w)]$ and $\int_{\mathcal{U}} U(w)\one\{w\in\wtilde{\mathcal W}\}\one\{w\not\in\calW_{\Delta}\}[d \wtilde P(w) - d P(w)]$. Then we have
\begin{align*}
\epsilon_{u} & = \int_{\mathcal V}U(w)\one\{w\in\mathcal W^{*}\}[d P(w) - d \wtilde P(w)] - \int_{\mathcal{U}}U(w)\one\{w\in\wtilde{\mathcal W}\}[d \wtilde P(w) - d P(w)]\\
&\ge\Delta\cdot\left[{\text{TV}}(P || \wtilde P ) - \int_{\mathcal{U}}\one\{w\in\wtilde{\mathcal W}\}\one\{w\in\calW_{\Delta}\}[d \wtilde P(w) - d P(w)] \right] \\
&\ge\Delta\cdot{\text{TV}}(P || \wtilde P ) - \Delta\cdot\int_{\wtilde{\mathcal W}}\one\{w\in\calW_{\Delta}\} d\wtilde P(w)\\
&\ge\frac{\Delta}{2}\cdot {\text{TV}}(P || \wtilde P ),
\end{align*}
in which



\begin{itemize}
\item the first inequality is due to $U(w)\le U(w^{*})$ for any $w\in\calW_{\Delta}$ and $w^{*}\in\mathcal W^{*}$ according to \pref{eq: w_best}, and $U(w^{*})-U(w)>\Delta$ for any $w\in\wtilde{\mathcal W}\setminus\calW_{\Delta}$ and $w^{*}\in\mathcal W^{*}$ from \pref{eq:l_U_bound_1}.
\item the second inequality is due to $\int_{\mathcal{U}}\one\{w\in\wtilde{\mathcal W}\}\one\{w\in\calW_{\Delta}\}[d \wtilde P(w) - d P(w)]\le\int_{\mathcal{U}}\one\{w\in\wtilde{\mathcal W}\}\one\{w\in\calW_{\Delta}\}d \wtilde P(w)\le \int_{\wtilde{\mathcal W}}\one\{w\in\calW_{\Delta}\}d \wtilde P(w)$. 
\item the third inequality is due to $\int_{\wtilde{\mathcal W}}\one\{w\in\calW_{\Delta}\} d\wtilde P(w)\le\frac{{\text{TV}}(P || \wtilde P )}{2}$.

\end{itemize}
\end{proof}
\section{Analysis of \pref{thm: utility-privacy trade-off_mt}}\label{sec: analysis_NFL_LLM}

With  \pref{lem: total_variation-utility trade-off_app} and \pref{lem: relation_utility_distortionpp}, it is now natural to provide a quantitative relationship between the utility loss and the privacy leakage (\pref{thm: utility-privacy trade-offpp}).

\begin{thm}[No free lunch theorem (NFL) for privacy and utility]\label{thm: utility-privacy trade-offpp} 
Let $\epsilon_p$ be defined in \pref{defi: privacy_leakage}, and let $\epsilon_u$ be defined in \pref{defi: utility_loss}, with \pref{assump: assump_of_Delta} we have that
\begin{align*}
    \frac{C_2}{C_1}\cdot\epsilon_p + \epsilon_u\ge C_2\cdot{\text{TV}}(P || \breve P ),
\end{align*}
where $C_1 = 1 - \frac{c_b + c_2\cdot c_b I^{p-1}}{\Omega}$ and $C_2 = \frac{\Delta}{2}$.
\end{thm}

\begin{proof}
From \pref{lem: total_variation-utility trade-off_app}, we have

\begin{align}\label{eq: privacy_lower_bound}
    \epsilon_{p} \ge C_1\cdot {\text{TV}}({\wtilde P} || \breve P )
\end{align}

Let $C_2 = \frac{\Delta}{2}$. From \pref{lem: relation_utility_distortionpp}, we have
\begin{align}\label{eq: utility_lower_bound}
    \epsilon_{u} \ge C_2\cdot {\text{TV}}(P || \wtilde P ).
\end{align}

Combining \pref{eq: privacy_lower_bound} and \pref{eq: utility_lower_bound}, we have that

\begin{align*}
    \frac{C_2}{C_1}\cdot\epsilon_p + \epsilon_u 
    & = C_2\cdot({\text{TV}}({\wtilde P} || \breve P ) + {\text{TV}}(P || \wtilde P )) \\
    &\ge C_2\cdot{\text{TV}}(P || \breve P ). 
\end{align*}

Therefore, we have that
\begin{align*}
    \frac{C_2}{C_1}\cdot\epsilon_p + \epsilon_u\ge C_2\cdot{\text{TV}}(P || \breve P ),
\end{align*}

where $C_1 = 1 - \frac{c_b + c_2\cdot c_b I^{p-1}}{\Omega}$ and $C_2 = \frac{\Delta}{2}$.


\end{proof}

\end{document}